\documentclass[journal]{IEEEtran}
\usepackage{amsmath,amssymb}

\usepackage{subcaption}
\usepackage{bm}
\usepackage{graphicx,graphics,color,psfrag}
\usepackage{cite,balance}
\usepackage{caption}
\allowdisplaybreaks
\usepackage{algorithm}
\usepackage{algorithmic}
\usepackage{accents}  
\usepackage{amsthm}
\usepackage{url}
\usepackage[english]{babel}
\usepackage{multirow}
\usepackage{enumerate}
\usepackage{cases}
\usepackage{stfloats}
\usepackage{dsfont}
\usepackage{color,soul}
\usepackage{amsfonts}
\usepackage{cite,graphicx,amsmath,amssymb}
\usepackage{fancyhdr}
\usepackage{hhline}
\usepackage{graphicx,graphics}
\usepackage{array,color}
\usepackage{mathtools}
\usepackage{amsmath}
\graphicspath{{./Figs/}}

\newtheorem{definition}{\emph{\underline{Definition}}}

\newtheorem{lemma}{\emph{\underline{Lemma}}}

\newtheorem{proposition}{\emph{\underline{Proposition}}}

\newtheorem{example}{\emph{\underline{Example}}}
\newtheorem{remark}{\bf \emph{\underline{Remark}}}

\def\l{\left}
\def\r{\right}
\def\({\left(}
\def\){\right)}

\def\b0{{\mathbf{0}}}

\begin{document}
	\captionsetup[figure]{name={Fig.}} 
	\title{MA-enhanced Mixed Near-field and Far-field Covert Communications}
	\author{Chao  Zhou, 
		Changsheng~You,~\IEEEmembership{Member,~IEEE},  Cong Zhou, Hai Lin,~\IEEEmembership{Senior Member,~IEEE},\\ and  Yi Gong,~\IEEEmembership{Senior Member,~IEEE}
		
		\thanks{
        Chao Zhou, Changsheng You,  Cong Zhou, and Yi Gong are with the Department of Electronic and Electrical Engineering, Southern University of Science and Technology (SUSTech), Shenzhen
		518055, China (e-mail: zhouchao2024@mail.sustech.edu.cn, youcs@sustech.edu.cn, zhoucong@stu.hit.edu.cn, and gongy@sustech.edu.cn). 
            
        Hai Lin is with the Graduate School of Engineering, Osaka Metropolitan University, Sakai, Osaka 599-8531, Japan (e-mail: hai.lin@ieee.org).
            
            \emph{Corresponding author: Changsheng You.}  
		}\vspace{-14pt}} 

	\maketitle
	\begin{abstract}
In this paper, we propose to employ a modular-based movable {\it{extremely large-scale array}} (XL-array) at Alice for enhancing covert communication performance. Compared with existing work that mostly considered either far-field or near-field covert communications, we consider in this paper a more general and practical {\it{mixed-field}} scenario, where multiple Bobs are located in either the near-field or far-field of Alice, in the presence of multiple near-field Willies. Specifically, we first consider a two-Bob-one-Willie system and show that conventional fixed-position XL-arrays suffer degraded sum-rate performance due to the {\it{energy-spread effect}} in mixed-field systems, which, however, can be greatly improved by subarray movement. On the other hand, for transmission covertness, it is revealed that sufficient angle difference between far-field Bob and Willie as well as adequate range difference between near-field Bob and Willie are necessary
for ensuring covertness in fixed-position XL-array systems, while this requirement can be relaxed in movable XL-array systems thanks to flexible channel correlation control between Bobs and Willie. 
Next, for general system setups, we formulate an optimization problem to maximize the achievable sum-rate under covertness constraint. To solve this non-convex optimization problem, we first decompose it into two subproblems, corresponding to an inner problem for beamforming optimization given positions of subarrays and an outer
problem for subarray movement optimization. Although these two subproblems are still non-convex, we obtain their high-quality solutions by using the successive convex approximation technique and devising a customized differential evolution algorithm, respectively. 
Last, numerical results demonstrate the effectiveness of proposed movable XL-array in balancing sum-rate and covert communication requirements, as compared to various benchmark schemes.
	\end{abstract}
\vspace{-6pt}    
\section{Introduction}
Extremely large-scale arrays (XL-arrays) have emerged as a promising and advanced technology to greatly enhance the spectral efficiency and spatial resolution of future wireless networks \cite{XLMIMO_tutorial,LiuTSBeam}. In particular, the huge number of antennas fundamentally alters the electromagnetic propagation characteristics, shifting from conventional far-field planar wavefronts to the  \emph{near-field} \emph{spherical} wavefronts~\cite{XLMIMO_tutorial}.
Such spherical wavefronts enable near-field \emph{beam-focusing}, for which the beam energy of XL-arrays can be concentrated at/around a targeted region instead of beam steering at certain angles. In addition, apart from the angle-domain beam manipulation, near-field beam-focusing allows more flexible beam control in both the angle and range domains, thus greatly enhancing the performance of various wireless applications, such as wireless power transfer~\cite{Zhangyp_swipt}, physical layer security (PLS)~\cite{LiuyuanweiPLS}, wireless sensing~\cite{LiuyuanweiSensing}, etc. In this paper, we study covert communication designs under a general and practical setup, namely, \emph{mixed} near-field and far-field communications. We show that the mixed-field channel correlation generally results in degraded covert rate performance, which, however, can be greatly improved by exploiting movable antenna (MA) technology.  

\vspace{-8pt}
\subsection{Prior Work}

\subsubsection{Covert communications} Covert communications have received upsurging research interests recently, driven by the practical needs for privacy/security protection. Compared with PLS that aims to avoid confidential information being decoded by eavesdroppers, covert communications target a more ambitious goal, that is, hiding communication behaviors in wireless channels and making data transmissions undetectable by unintended users \cite{chen2023covert}. For transceiver designs, most  existing work on covert communications considered either the far-field or near-field systems. Specifically, for far-field covert communications, multiple-input multiple-output (MIMO) technology has been widely applied to form pencil-like beams based on acquired channel state information (CSI), hence effectively reducing signal power leakage and improving transmission covertness~\cite{Multi-antCovert}.  In addition, the authors in \cite{Muiti-antenna_CC} proposed to design an interference beam (i.e., artificial noise (AN)) at the jammer towards Willie so as to confuse data detection at Willie, while ensuring the legitimate user (Bob) located in its null space. However, these methods may not be efficient when there is strong channel correlation between Bobs and Willies, e.g., locating at the same angle. To tackle this difficulty, one promising approach is leveraging intelligent reflecting surfaces (IRSs) for smartly controlling signal reflections, so as to enhance received power at intended Bobs as well as reduce signal power at undesired Willies \cite{NgIRS_covert}. 
Alternatively, reconfigurable antenna technologies, such as MAs~\cite{Zhu_MA_mag,YouNGAT,ShaoJSAC_6DMA,ShaoSurvey}, rotatable antennas~(RAs)\cite{zheng2025rotatable}, and fluid antennas (FAs)~\cite{WongFAS}, offer another option to address this issue effectively. Specifically, by leveraging their unique spatial degrees-of-freedom (DoFs), these technologies have shown significant potential in reducing spatial channel correlation to achieve covertness. For example, the authors in~\cite{MA_covert} revealed the capability of MAs to evade detection by multiple wardens in covert communication systems.

For near-field systems, the authors in~\cite{RIS_NFC_CC} proposed to exploit the near-field beam-focusing property for reducing power and information leakage to Willie, even when Bob and Willie are located at the same angle. This approach offers new opportunities for designing covert communications in the range domain. In addition, a vulnerable region for near-field covert communication systems was characterized in~\cite{NCC_CCregion},  which is determined by both the angle and range of the legitimate user. Moreover, by exploiting spherical wavefronts, the authors in~\cite{lotfi2024covertness} proposed to employ a frequency diverse array (FDA) in near-field covert communication systems to achieve flexible beampatterns, which enhances beam-focusing performance for the legitimate user while reducing power leakage towards Willie. This work was further extended  in~\cite{zhang2023robust}, where a robust beam-focusing design was developed for a  practical scenario accounting for uncertain user locations.
\subsubsection{Mixed-field wireless systems}
In practice, far-field or near-field users may not exist \emph{alone} in wireless systems. For example, when a base station (BS) equipped with 256 antennas operates at $30$ GHz frequency band, its Rayleigh distance (defined based on phase variations across antennas) is about 325 meters (m), while its \emph{effective Rayleigh distance}~\cite{Cui_ERaydis} (defined based on array gain variations) is about 120 m. Note that compared with classic Rayleigh distance, effective Rayleigh distance is generally more reasonable in evaluating rate performance~\cite{Cui_ERaydis}, since it is directly related to array gains. As such, for typical wireless systems with a cell radius of 200 m, there may exist both near-field and far-field users concurrently. For such mixed-field wireless systems, it was revealed in~\cite{Zhangyp_mixed} that the discrete Fourier transform (DFT) beams towards far-field users may introduce significant interference to near-field users, even when they are located at different angles. This is fundamentally due to the \emph{non-orthogonality} between near-field and far-field channel steering vectors, which leads to far-field power leakage in the near-field region, referred to as the \emph{energy-spread effect}~\cite{WuDFT}. This work was further extended in~\cite{Zhangyp_swipt}, where the authors leveraged far-field beams to power near-field energy harvesting receivers by exploiting the energy-spread effect. Moreover, similar mixed-field setups were considered in~\cite{zhou2025mixed,LichenmixedISAC} for a variety of applications such as mixed-field target localization~\cite{zhou2025mixed} and mixed-field integrated sensing and communications (ISAC)~\cite{LichenmixedISAC}. Nevertheless, for covert communications, mixed-field correlation may affect both the achievable rate (e.g., between near-field and far-field Bobs) and transmission covertness (e.g., between near-field Willie and far-field Bobs). How to balance these two performance in mixed-field channels remains an open problem.

\vspace{-6pt}
\subsection{Motivations and Contributions}
Motivated by the above, we consider in this paper a practical mixed-field covert communication system as shown in Fig.~\ref{Fig:System_Model}, where a multi-antenna Alice transmits covert information to multiple near-field and far-field Bobs, in the presence of multiple near-field Willies.  To enhance covert communication performance, we propose to employ a \emph{modular}-based \emph{movable} XL-array~\cite{YZmodular1}, whose antennas are divided into a set of movable modules/subarrays. Note that conventional antenna-based MA usually faces practical challenges such as demanding hardware complexity to control the movement of all antennas and high computational complexity in the algorithm designs. In contrast, the considered movable subarrays (or array-based MA) greatly reduce both the hardware and computational complexity. Moreover, as compared to conventional XL-arrays without antenna movement (named as \emph{fixed-position} XL-array, shortened as fixed XL-array), the considered movable XL-array provides an additional spatial DoF to deal with the energy-spread effect.  

To our best knowledge, this work represents the \emph{first} attempt to study covert communication designs under mixed-field channels. The main contributions of this paper are summarized as follows.
\begin{itemize}
\item First, to shed useful insights, we consider a typical system consisting of a near-field Bob, a far-field Bob, and a near-field Willie. Interestingly, we show that the mixed-field energy-spread effect significantly reduces the achievable sum-rate of fixed XL-array systems, while movable subarrays can effectively reduce the near-far channel correlation and hence substantially improve achievable sum-rate. In addition, for transmission covertness, we characterize the covert transmission region for far-field Bob in the angle domain and that for near-field Bob in the range domain, respectively. It is revealed that sufficient angle difference between far-field Bob and Willie as well as adequate range difference between near-field Bob and Willie are necessary for ensuring covertness in fixed XL-array systems. However, this issue can be greatly addressed by subarray movement optimization in movable XL-array systems, which effectively reconfigure mixed-field channel correlation.

\item Next, for general system setups, we formulate an optimization problem to maximize the achievable sum-rate under the covertness constraint. To solve this non-convex optimization problem, we decompose it into two subproblems, namely, an inner problem for beamforming optimization given positions of subarrays and an outer problem for subarray movement optimization. Then, the inner problem is efficiently solved by using the successive convex approximation (SCA) technique and the outer problem is solved by a customized differential evolution (DE) algorithm.  
\item 
Last, numerical results are presented to demonstrate the effectiveness of the proposed movable XL-array in enhancing mixed-field covert communication performance.
It is shown that the movable XL-array significantly reduces inter-user interference caused by the mixed-field energy-spread effect, leading to noticeable improvement in the achievable sum-rate as compared to the fixed XL-array. 
In addition, by deploying a movable XL-array, the strong mixed-field channel correlations between Bobs and Willies are effectively suppressed, hence leading to a substantial enhancement in covert rate performance.
\end{itemize}

 
	
    \emph{Notations:} 
       We use lowercase boldface letters (e.g., $\mathbf{a}$) to denote vectors, uppercase boldface letters (e.g., $\mathbf{A}$) for matrices, and calligraphic letters (e.g., $\mathcal{M}$) to represent sets. The superscripts $(\cdot)^H$ and $(\cdot)^T$  indicate Hermitian (conjugate) transpose and regular transpose operations, respectively.  The union operation is denoted by $\cup$. The notation $|\cdot|$ represents the absolute value of a scalar and the cardinality of a set. The notation $||\cdot||_F$ represents the Frobenius norm of a matrix. Complex Gaussian distribution is denoted by $\mathcal{CN}(\mu,\sigma^2)$ and $\mathcal{O}(\cdot)$ denotes the standard big-O notation. 
  \vspace{-0.3cm}
  \section{System Model}\label{Sec2:label}
	We consider an MA-enhanced mixed-field covert communication system as shown in Fig.~\ref{Fig:System_Model}, where a multi-antenna Alice (or BS) transmits covert information to $B$ single-antenna Bobs (denoted by $\mathcal{B}\triangleq\{1,2,\ldots,B\}$), in the presence of $W$ single-antenna Willies (denoted by $\mathcal{W}\triangleq\{1,2,\ldots,W\}$).  Specifically, Alice is equipped with a modular-based movable XL-array, for which the XL-array of $N$ antennas (denoted  by $\mathcal{N} \triangleq \{1, 2, \ldots, N\}$) is divided into $M$ movable subarrays (denoted by $\mathcal{M} \triangleq \{1, 2, \ldots, M\}$), each consisting of $\tilde{N}=N/M$ antennas. Moreover, we consider a practical and challenging \emph{mixed-field} scenario, where Bobs are located in either the near-field or far-field of Alice, while Willies are located in the near-field for monitoring communication signals at short distances. 
	
	\begin{figure}[t]
		\centering
		\includegraphics[width=0.48\textwidth]{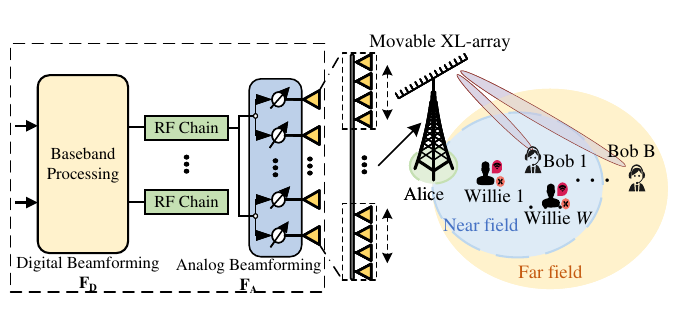}
		\caption{The considered modular-based MA-enhanced mixed-field covert communication system.} \label{Fig:System_Model}
		\vspace{-20pt}
	\end{figure}
 \vspace{-6pt}   
\subsection{XL-array Model}
For the considered modular-based movable XL-array, the inter-antenna spacing in each subarray is half-wavelength, while the position of each subarray can be flexibly controlled to enhance covert communication performance. Without loss of generality, the movable XL-array is positioned at the ${y}$-axis. Let $\mathbf{q}=[{q}_{1},{q}_{2},\ldots,{q}_{M}]^T \in \mathbb{R}^{M \times 1}$ denote the position vector of $M$ subarrays, where ${q}_{m}, \forall m\in \mathcal{M}$ is the central position of subarray $m$. Given half-wavelength inter-antenna spacing in each subarray, the position of the $n$-th antenna of the XL-array (or the $\tilde{n}$-th antenna of subarray $m$ with $n = \tilde{n} + (m - 1)\tilde{N}$) can be expressed as   
\begin{align}\label{Exp:antenna_position}
		\hat{q}_{n} = {q}_{m} + \frac{2\tilde{n}-\tilde{N}-1}{2} d, \forall n\in \mathcal{N}.
	\end{align}
Moreover, the central positions of subarrays $i$ and $j$ satisfy\footnote{In this paper, to characterize the performance upper bound, we consider the continuous MA array movement, while the proposed design can be further extended to the case of discrete movement by using e.g., mapping and rounding methods. In addition, in practice, the movement resolution of MA at $27.5$ GHz frequency band can achieve  $0.05\lambda$, leading to comparable rate performance with the continuous case~\cite{dong2024movable}.} $$|{q}_{i} - {q}_{j}| \geq \tilde{N}d,~~~ 1 \leq i \neq j \leq M.$$  In particular, when $|{q}_{m+1}-{q}_{m}| = \tilde{N}d, \forall m \in \mathcal{M}$, the movable XL-array reduces to the fixed XL-array, for which the inter-antenna spacing of adjacent antennas is half-wavelength. For fixed XL-array, we denote by $\mathbf{q} =  \mathbf{q}_{\rm U}$ the central positions of its subarrays. 
	
\vspace{-6pt}	
\subsection{Channel Models}
\vspace{-2pt}	
For the boundary between the near-field and far-field region, we adopt the \emph{effective Rayleigh distance} in this paper, which is defined as $Z(\theta)  = \nu \frac{2D^2(1-\theta^2)}{\lambda}$ with $\nu = 0.37$~\cite{Cui_ERaydis} and $D$ denoting the XL-array aperture.
It is worthy noting that compared with the classic Rayleigh distance obtained based on phase variations, the effective $Z(\theta)$ characterizes the distance boundary where the array gains obtained from the far-field and near-field channel models exhibit notable differences. 
	
For both near-field and far-field Bobs, we adopt the generic near-field spherical wavefront multi-path channel model. Let $\mathbf{h}^H_{{\rm B},b}(\mathbf{q}) \in  \mathbb{C}^{1 \times N}$ denote the channel from Alice to Bob $b$, which is modeled as 	\begin{align}\label{Exp:Bobchannel}
		\mathbf{h}^H_{{\rm B},b}(\mathbf{q}) &= \sqrt{N} g_{{\rm B},b} \mathbf{a}^H(\mathbf{q},\theta_{{\rm B},b}, r_{{\rm B},b}) \nonumber \\
		& +\sqrt{\frac{N}{L_{b}}}\sum_{\ell_{b}=1}^{L_{b}}  g_{{\rm B},b}^{(\ell_{b})} \mathbf{a}^H(\mathbf{q},\theta_{{\rm B},b}^{(\ell_{b})}, r_{{\rm B},b}^{(\ell_{b})} ), \forall b \in \mathcal{B},
	\end{align}
	which consists of one line-of-sight (LoS) path and $L_{b}$ non-LoS (NLoS) paths. Additionally, $g_{{\rm B},b}$ and $ g_{{\rm B},b}^{(\ell_{b})} $ are the complex-valued path gains of the LoS path and the $\ell_{b}$-th NLoS path with $ \{\theta_{{\rm B},b},r_{{\rm B},b}\} $ and $\{\theta_{{\rm B},b}^{(\ell_{b})},r_{{\rm B},b}^{(\ell_{b})}\}$ denoting their corresponding spatial angles and ranges, respectively. Herein, the channel steering vector (i.e., $\mathbf{a}({\mathbf{q},\theta,r})\in \mathbb{C}^{N \times 1 }$) is given by
	\begin{align}\label{Exp:steer_vec}
		\mathbf{a}(\mathbf{q},\theta, r) = \frac{1}{\sqrt{N}} [\mathbf{b}_{1}^T\left({q}_{1},\theta,r\right),
		\ldots,
		\mathbf{b}_{M}^T\left({q}_{M},\theta,r \right)]^T,
	\end{align}
	where $ \mathbf{b}_{m}\left({q}_{m},\theta,r \right) \in \mathbb{C}^{\tilde{N} \times 1} $ is the near-field channel steering vector of the $m$-th  subarray given $\{{q}_{m},\theta,r\}$ with $\theta \in [-1,1]$. Specifically,  $ \mathbf{b}_{m}\left({q}_{m},\theta,r \right) $ is given by
	\begin{align}\label{Exp:Sub_steer_vec}
		\mathbf{b}_{m}\left({q}_{m},\theta,r \right) = \Big[e^{-\jmath\frac{2\pi}{\lambda}({q}_{m,1} \theta - \frac{{q}_{m,1}^2 (1-\theta^2) }{2 r} )}, \ldots,\nonumber \\
		e^{-\jmath\frac{2\pi}{\lambda}({q}_{m,\tilde{N}} \theta - \frac{{q}_{m,\tilde{N}}^2 (1-\theta^2)}{2 r} )}\Big]^T.
	\end{align}
Note that for far-field Bobs, as $r>Z(\theta) $, the near-field channel steering vector in~\eqref{Exp:steer_vec} reduces to the far-field case, which is simply denoted as $\hat{\mathbf{a}} 
(\mathbf{q},\theta)$ by ignoring the quadratic phase term in~\eqref{Exp:Sub_steer_vec}.

We consider scenarios where both Bobs and Willies are licensed users in the system, while Alice suspects Willies as malicious users during data transmissions~\cite{RIS_NFC_CC}.\footnote{The proposed design also applies to the case where Willies are unlicensed users, while Alice can obtain the CSI of Willies by using e.g., advanced detecting equipment like ``Ghostbuster” introduced in~\cite{Chaman2018Ghostbuster}.}
Moreover,  we assume that the CSI of both Bobs and Willies is available at Alice by using existing channel estimation methods (see, e.g., \cite{Cui2022CE,YouNGAT}).\footnote{For scenarios with imperfect CSI of Willies, we present the impact of CSI errors on covert transmission probability in Section~\ref{Sec:Sim}.
Note that for scenarios involving imperfect CSI of Willies and/or Bobs, the robust beamforming design can be adopted to address the CSI uncertainty~\cite{MyTWC}; however, this lies beyond the scope of this work and is left for our future work.}
Similarly, the channel from Alice to the $w$-th  Willie, denoted by $\mathbf{h}^H_{{\rm W},w}(\mathbf{q})\in \mathbb{C}^{1 \times N} $, is modeled as
	\begin{align}\label{Exp:Williechannel}
		\mathbf{h}^H_{{\rm W},w}(\mathbf{q}) &= \sqrt{N} g_{{\rm W},w} \mathbf{a}^H(\mathbf{q},\theta_{{\rm W},w}, r_{{\rm W},w}) \nonumber \\
		+& \sqrt{\frac{N}{L_{w}}}\sum_{\ell_{w}=1}^{L_{w}}  g_{{\rm W},w}^{(\ell_{w})} \mathbf{a}^H(\mathbf{q},\theta_{{\rm W},w}^{(\ell_{w})}, r_{{\rm W},w}^{(\ell_{w})} ), \forall w\in\mathcal{W},
	\end{align}
	where $ g_{{\rm W},w} $, $ \theta_{{\rm W},w} $, and $ r_{{\rm W},w} $ denote the complex-valued path gain, spatial angle, and range of the LoS path. Moreover, $ g_{{\rm W},w}^{(\ell_{w})} $, $ \theta_{{\rm W},w}^{(\ell_{w})} $, and $ r_{{\rm W},w}^{(\ell_{w})} $ are the complex-valued path gain, spatial angle, and range of the $\ell_{w}$-th NLoS path with $L_{w}$ representing the number of NLoS paths for Willie $w$.
    Additionally, $ \mathbf{a}(\mathbf{q},\theta_{{\rm W},w}, r_{{\rm W},w}) $ and $\mathbf{a}(\mathbf{q},\theta_{{\rm W},w}^{(\ell_{w})}, r_{{\rm W},w}^{(\ell_{w})} ) $ represent the steering vectors, which are given in~\eqref{Exp:steer_vec}.
	
    \vspace{-8pt}	
    \subsection{Signal Model}
	We consider a low-cost hybrid beamforming architecture for the movable XL-array, where $B$ ($ B \ll N $) radio frequency (RF) chains are equipped at the BS to serve $B$ Bobs. Let $ \mathbf{F}_{\rm A} \in \mathbb{C}^{ N \times B}$ and $ \mathbf{F}_{\rm D} \in \mathbb{C}^{B \times B}$ represent the analog and digital beamforming matrices, respectively. Based on the above, the received signal at the $b$-th Bob is given by 
	\begin{align}
		y_{{\rm B},b} = \mathbf{h}^H_{{\rm B},b}(\mathbf{q}) \mathbf{F}_{\rm A}  \mathbf{F}_{\rm D} \mathbf{s} + z_{{\rm B},b},~\forall b\in \mathcal{B},
	\end{align}
	where $\mathbf{s} = [s_{1},\ldots,s_{B}]^T \in \mathbb{C}^{B \times 1}$ is the transmitted symbol vector for the $B$ Bobs and $ z_{{\rm B},b} \sim \mathcal{CN}(0,\sigma_{{\rm B}}^2)$ is the additive white Gaussian noise (AWGN) at Bob $b$. As such, the achievable rate in bits/second/hertz (bps/Hz) at Bob $b$ can be expressed as 
	\begin{align}
		& R_{{\rm B},b}(\mathbf{q},\mathbf{F}_{\rm A},\mathbf{F}_{\rm D})  \nonumber  \\
		&= \log_2\l(  1 + \frac{ |\mathbf{h}^H_{{\rm B},b}(\mathbf{q}) \mathbf{F}_{\rm A} \mathbf{f}_{{\rm D},b}   |^2}{ \sum_{i=1,i\neq b}^{B} |\mathbf{h}^H_{{\rm B},b}(\mathbf{q}) \mathbf{F}_{\rm A} \mathbf{f}_{{\rm D},i}   |^2 + \sigma_{{\rm B}}^2  } \r) , \forall b\in \mathcal{B}, \nonumber
	\end{align}
	where $ \mathbf{f}_{{\rm D},b} $ is the $b$-th column of $\mathbf{F}_{\rm D}$.
    
\vspace{-6pt}	
\subsection{Hypothesis Test and Detection Performance}
For covert communications, each Willie aims to detect the existence of communications between Alice and Bobs based on the Neyman-Pearson test~\cite{chen2023covert} independently, which can be formulated as a binary hypothesis testing problem. Specifically, for each Willie $w$, its  received signal at time slot $t$ is given by 
	\begin{equation}
		y_{{\rm W},w}(t)  =\left\{
		\begin{aligned}
			&z_{{\rm W},w}(t), &&\mathcal{H}_{0}, \\
			&\mathbf{h}^H_{{\rm W},w}(\mathbf{q}) \mathbf{F}_{\rm A}  \mathbf{F}_{\rm D} \mathbf{s}
			+ z_{{\rm W},w}(t),&&\mathcal{H}_{1},
		\end{aligned}
		\right.
	\end{equation} 
	where $ z_{{\rm W},w}(t) \sim \mathcal{CN}(0,\sigma_{\rm w}^{2})$ is the AWGN at Willie $w$. As energy detection technique is utilized by Willies to determine covert transmissions between Alice and Bobs, noise uncertainty needs to be considered, which affects the covertness performance. Similar to~\cite{Si2021Covertcom,STAR_Covert}, we consider the bounded uncertainty model, based on which the probability density function (PDF) of noise power $\sigma_{\rm w}^2$ is given by
\begin{equation}\label{Exp:noise_PDF}
		f_{\sigma_{\rm w}^2}(x) = \left\{
		\begin{aligned}
			&\frac{1}{2\ln(\rho)x },   && \textrm{if}~\frac{1}{\rho} \tilde{\sigma}_{\rm w}^2\le x \le \rho \tilde{\sigma}_{\rm w}^2, \\
			&0,   && \textrm{otherwise},
		\end{aligned} 
		\right.
	\end{equation}
	where $ \rho \ge 1 $ represents the noise uncertainty level and $\tilde{\sigma}_{\rm w}^2$ denotes the nominal noise power~\cite{STAR_Covert}. Based on the above, the average power of received signal at Willie $w$ is given by
	\begin{align}
		\bar{P}_{{\rm W},w} &= \frac{1}{T}\sum_{t=1}^{T} |y_{{\rm W},w}(t)|^2 \nonumber \\
		&=\left\{ 
		\begin{aligned}
			&~{\sigma}_{\rm w}^2, 	&& \mathcal{H}_{0}, \\
			&~\sum_{b=1}^{B}|\mathbf{h}^H_{{\rm W},w}(\mathbf{q}) \mathbf{F}_{\rm A} \mathbf{f}_{{\rm D},b} |^2   + {\sigma}_{\rm w}^2,   && \mathcal{H}_{1}.
		\end{aligned}
		\right.
	\end{align}
	For each Willie $w$, the power detection threshold $\tau$ is properly set to determine the existence of covert transmission ($\mathcal{D}_{1}$) or not ($\mathcal{D}_{0}$),  which is given by
	\begin{align}
		\mathcal{D}_{0}:\bar{P}_{{\rm W},w} \le \tau,~~\text{and}~~
		\mathcal{D}_{1}:\bar{P}_{{\rm W},w} \ge \tau.  
	\end{align}
 In the above energy detection, two kinds of errors may occur, namely, false alarm (FA) and miss detection (MD), whose probabilities are respectively given~by 
	\begin{subequations}
		\begin{align}
			\epsilon_{{\rm FA},w}(\tau) & \triangleq \text{Pr}\{\mathcal{D}_{1}|\mathcal{H}_{0}\} = \text{Pr}\{ {\sigma}_{\rm w}^2 \ge \tau \},\\
			\epsilon_{{\rm MD},w}(\tau) &\triangleq \text{Pr}\{\mathcal{D}_{0}|\mathcal{H}_{1}\} 
			= \text{Pr}\{f_{{\rm Co},w} + {\sigma}_{\rm w}^2 \le \tau \},
		\end{align}
	\end{subequations}
	where $f_{{\rm Co},w} \triangleq \sum_{b=1}^{B}|\mathbf{h}^H_{{\rm W},w}(\mathbf{q}) \mathbf{F}_{\rm A} \mathbf{f}_{{\rm D},b} |^2$ represents the average received power of information-bearing signals at Willie $w$. As such, the total detection error probability (DEP) of Willie $w$ given parameter $\tau$ is 
	\begin{align}\label{Exp:DEP}
		\epsilon_{w}(\tau) = \epsilon_{{\rm FA},w}(\tau) + \epsilon_{{\rm MD},w}(\tau), \forall w\in \mathcal{W}.
	\end{align}
	Based on the first-order optimality condition~\cite{Si2021Covertcom},  the optimal detection threshold $\tau_{w}^{*}$ at Willie $w$ to minimize the DEP should be set as
	\begin{align}
		\tau_{w}^{*} =\min\Big\{  f_{{\rm Co},w} +  \frac{1}{\rho}\tilde{\sigma}_{\rm w}^2,\rho \tilde{\sigma}_{\rm w}^2 \Big\}, \forall w\in \mathcal{W}.
	\end{align}
By substituting $\tau_{w}^{*}$ and~\eqref{Exp:noise_PDF} into~\eqref{Exp:DEP}, the minimum DEP for Willie $w$ can be obtained as
	\begin{equation}\label{Exp:MDEP}
		\epsilon_{w}^{*} = \left\{
		\begin{aligned}
			&1-\frac{1}{2\ln(\rho)}\ln\!\left(\! 1+\frac{\rho f_{{\rm Co},w}}{\tilde{\sigma}_{\rm w}^2}\!\right) ,~\textrm{if}~f_{{\rm Co},w} \le \tilde{\sigma}_{\rm w}^2 (\rho - \frac{1}{\rho}), \\
			&0,~~~~~~~~~~~~~~~~~~~~~~~~~~~~~~~~~~\textrm{otherwise}.
		\end{aligned} 
		\right.
	\end{equation}
Let $ \varsigma \in [0, 1] $ denote a positive constant, which is employed to ensure covert communication requirement for the considered system. Specifically, covert communication is guaranteed only when the minimum DEPs of all  Willies satisfy $ \epsilon_{w}^{*} \geq 1 - \varsigma, \forall w \in \mathcal{W} $. This condition can be equivalently rewritten as
\begin{align}\label{Exp:CT_constraint}
		f_{{\rm Co},w}  
		&\le \min \left\lbrace \tilde{\sigma}_{\rm w}^2 \left( \rho - \frac{1}{\rho}\right) ,\tilde{\sigma}_{\rm w}^2 \left(   \frac{\rho^{2\varsigma} - 1}{\rho}\right)  \right\rbrace  \nonumber \\
		& = \tilde{\sigma}_{\rm w}^2 \left(   \frac{\rho^{2\varsigma} - 1}{\rho}\right) ,\forall w \in \mathcal{W}.
	\end{align} 
	Note that the inequality in \eqref{Exp:CT_constraint} indicates that the average received power of information-bearing signals at each Willie should be smaller than a certain threshold to ensure covertness.

	\section{Covert Communication Performance in MA-enhanced Mixed-field Systems}\label{Sec3}
	To draw useful insights into the system performance with and without subarray movement, we consider a typical mixed-field scenario where Alice communicates with a near-field Bob and a far-field Bob, while a near-field Willie detects the covert transmissions between Alice and Bobs.  Interestingly, we show that movable XL-array can effectively deal with the energy-spread effect in mixed-field scenarios and hence greatly improve the achievable sum-rate as well as covertness conditions compared to fixed XL-array.

To this end, we first make the following definitions.
	\begin{definition}[Correlation function]\label{Def:1}
		\rm  
		For movable XL-array, given two locations $\{\theta_i,r_i\}$ and $\{\theta_{j},r_{j}\}$,  the correlation function of two steering vectors given by~\eqref{Exp:steer_vec} is defined as 
		\begin{align}
			&\chi(\mathbf{q},\{\theta_{i},r_{i}\},\{\theta_{j},r_{j}\}) 
			\triangleq \big|\mathbf{a}^H(\mathbf{q},\theta_{i},r_{i})  \mathbf{a}(\mathbf{q},\theta_{j},r_{j}) \big| \nonumber \\ 
			& = \frac{1}{N}  \Big| \sum_{n=1}^{N} e^{\jmath\frac{2\pi}{\lambda}\hat{q}_{n} (\theta_{j} - \theta_{i} ) + \jmath\frac{2\pi}{\lambda}\hat{q}_n^2\l(\frac{1-\theta_{i}^2}{2 r_{i}} -\frac{1-\theta_{j}^2}{2 r_{j}} \r) }
			\Big| \label{Exp:Corfun}. 
		\end{align}
	\end{definition}
	
	\begin{definition}\label{Def:2}
		\rm 
		For fixed XL-array (i.e., $\mathbf{q} =  \mathbf{q}_{\rm U}$), given $ \theta_{i} = \theta_{j} =\theta $, the correlation function in~\eqref{Exp:Corfun} can be approximated as~\cite{Cui2022CE}
		\begin{align}
			\hat{\chi}(\mathbf{q}_{\rm U},\{\theta,r_{i}\},\{\theta,r_{j}\})
			\approx  |G(\beta)| \triangleq \left| \frac{C(\beta) + \jmath S(\beta)}{\beta} \right|, 
		\end{align}
		where $\beta = \sqrt{\frac{N^2d^2 (1-\theta^2) }{2\lambda}\left|\frac{1}{r_{i}}-\frac{1}{r_{j}} \right| }$, $C(\beta)=\int_{0}^{\beta} \cos(\frac{\pi}{2}t^2)dt$,  and $S(\beta)=\int_{0}^{\beta} \sin(\frac{\pi}{2}t^2)dt$.
	\end{definition}
	\begin{definition}[$\mu$-dB energy-spread angular support]\label{Def:3}
		\rm 
		For fixed XL-array, given near-field channel steering vector $\mathbf{a}(\mathbf{q}_{\rm U},\theta_{i},r_{i}) $ and far-field beamforming vector $\mathbf{w}(\theta) = \hat{\mathbf{a}}(\mathbf{q}_{\rm U},\theta)$,  the $\mu$-dB energy-spread angular support is defined as~\cite{Zhang2022Fast}
		\begin{align}
			\mathcal{A}_{\mu{\rm dB}} (\theta_{i},r_{i})=&\Big\{\theta~|~
			f\left( \mathbf{w}\left( \theta \right) ,\theta_{i},r_{i}\right)   \nonumber \\
			&~~~> \iota \max_{\theta\in[-1,1]}  f\left( \mathbf{w}\left( \theta\right) ,\theta_{i},r_{i}\right) \Big\}, 
		\end{align}
         where  $f\left( \mathbf{w}\left( \theta \right) ,\theta_{i},r_{i}\right)  \triangleq | \mathbf{a}^H(\mathbf{q}_{\rm U},\theta_{i},r_{i}) \mathbf{w}(\theta)  |$ denotes the normalized beam gain and $\iota \!=\! 10^{-\frac{\mu}{10}}$. In particular, $\mu \!= \!10$ corresponds to the 10-dB energy-spread region with ${f_{\rm peak}}\triangleq \max_{\theta\in[-1,1]}  f\left( \mathbf{w}\left( \theta\right) ,\theta_{i},r_{i}\right)$ denoting the peak normalized beam gain at range $r_{j}$ (equivalently ${{0.1}f_{\rm peak}}$).
		Specifically, $\mathcal{A}_{\mu{\rm dB}} (\theta_{i},r_{i})$ characterizes an angular support for which the corresponding normalized beam gain is larger than $\iota {f_{\rm peak}}$.
	\end{definition}

    Note that the defined $\mu$-dB angular support~\cite{Zhang2022Fast} has been recently employed to characterize the energy-spread effect for facilitating mixed-field system performance analysis and DFT-based beam training design. For illustration, we plot in Fig.~\ref{Fig:Sec3-EnergySpread} the normalized beam gain with different $\{\theta_{i},r_{i}\}$ under far-field beamformers $\mathbf{w}(\theta)$ versus different spatial angles $\theta$. It is observed that the normalized beam gain has a significantly large value with slight fluctuations within a specific interval, which is generally enlarged when the user range is reduced.
	
	\begin{figure}[t]
		\centering
		\includegraphics[width=0.35\textwidth]{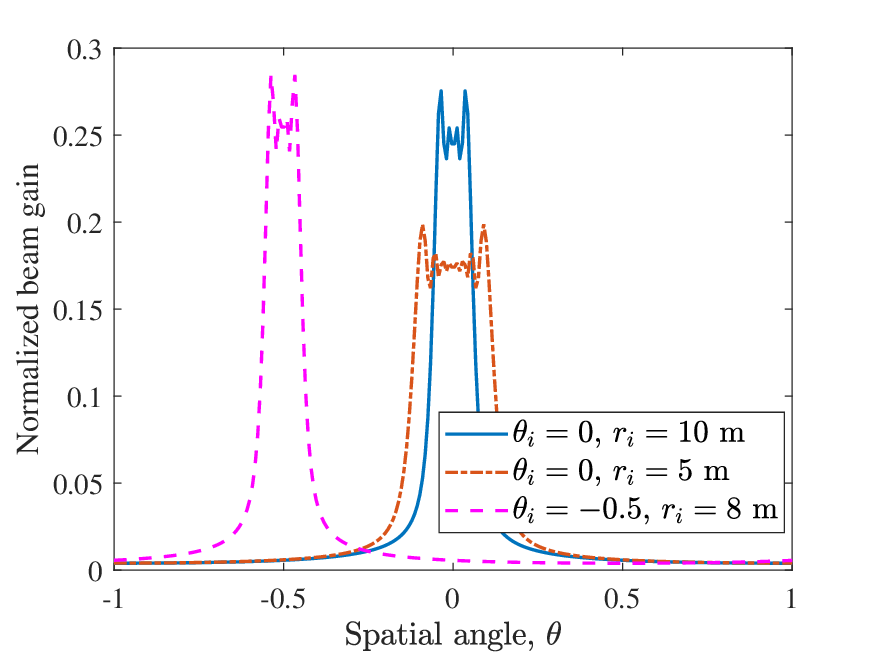}
		\caption{Normalized beam gain with far-field beamformers under different \{$\theta_{i},r_{i}$\} when $N=256$ and $f=30$ GHz.} \label{Fig:Sec3-EnergySpread}
		\vspace{-14pt}
	\end{figure}
	
	\subsection{Achievable Sum-rate Analysis}\label{Sec:III-A}
	\subsubsection{\textbf{\underline{Fixed XL-array}}}
	First, we study the achievable sum-rate performance for fixed XL-arrays when covertness constraints are not considered.
	To facilitate rate performance analysis under mixed-field channels, we consider the LoS-only scenario for high-frequency bands and a low-complexity hybrid beamforming design, while the more general multi-path scenario is considered in Section~\ref{Sec:IV}.
	Specifically, the analog beamforming is designed based on maximal ratio transmission (MRT) (i.e., steering beams towards  Bobs), while power allocation is optimized in digital beamforming~\cite{LiuTSBeam,liuty2025PLS,Wangzl}.\footnote{Note that more complicated digital beamforming designs (e.g., zero-forcing (ZF) or minimum mean-squared error (MMSE)) generally make the rate performance analysis more complicated and even intractable. Thus, we consider power allocation design in the performance analysis while digital beamforming will be optimized in Section~\ref{Sec:IV}.} 
As such,  the beamforming vectors are given by
	$\mathbf{F}_{\rm A} \mathbf{f}_{{\rm D},1} =\sqrt{P_{{\rm B},1}}\mathbf{a}({\mathbf{q}_{\rm U},\theta_{{\rm B},1},r_{{\rm B},1}})$ and $ \mathbf{F}_{\rm A} \mathbf{f}_{{\rm D},2} =\sqrt{P_{{\rm B},2}}\mathbf{a}({\mathbf{q}_{\rm U},\theta_{{\rm B},2},r_{{\rm B},2}}) $, where $P_{{\rm B},1}$ and $P_{{\rm B},2}$ denote the allocated transmit power of Alice to the two Bobs, respectively, with $P_{{\rm B},1} + P_{{\rm B},2} =P$. Based on the above, the achievable rates at the near-field Bob~$1$ and far-field Bob $2$ can be respectively rewritten as
	\begin{subequations}\label{Exp:Rate}
		\begin{align}
			R_{{\rm B},1} &= \log_2\l(1+\frac{P_{{\rm B},1} N g_{{\rm B},1}^2}{P_{{\rm B},2} N g_{{\rm B},1}^2 (\chi_{{{\rm B}_1},{{\rm B}_2}  } )^2  +   \sigma_{{\rm B}}^2}\r),\label{Exp:R1} \\
			R_{{\rm B},2} &= \log_2\l(1+\frac{P_{{\rm B},2} N g_{{\rm B},2}^2}{P_{{\rm B},1} N g_{{\rm B},2}^2 (\chi_{{{\rm B}_1},{{\rm B}_2}  } )^2  +   \sigma_{{\rm B}}^2}\r),\label{Exp:R2}
		\end{align} 
	\end{subequations}
	where $\chi_{{{\rm B}_1},{{\rm B}_2}  } $ represents the correlation between the channel steering vectors of Bob $1$ and Bob $2$. In addition, based on~\textbf{Definition~\ref{Def:3}},
	$\chi_{{{\rm B}_1},{{\rm B}_2}  } $ can be approximated as
	\begin{align}
			\chi_{{{\rm B}_1},{{\rm B}_2}  }  & =|\mathbf{a}^H({\mathbf{q}_{\rm U},\theta_{{\rm B},1},r_{{\rm B},1}}) \mathbf{a}({\mathbf{q}_{\rm U},\theta_{{\rm B},2},r_{{\rm B},2}})|   \\  
			&\approx \left\{ 
			\begin{aligned}
				&\chi(\mathbf{q}_{\rm U},\{\theta_{{\rm B},1},r_{{\rm B},1}\},\{\theta_{{\rm B},2},r_{{\rm B},2}\}), \\
				&~~~~~~~~~~~~~~~~~~~~~~~~\textrm{if}~\theta_{{\rm B},2} \in \mathcal{A}_{10{\rm dB}} (\theta_{{\rm B},1},r_{{\rm B},1}),\\
				&0,~~~~~~~~~~~~~~~~~~~~~\textrm{otherwise},
			\end{aligned} 
			\right. \nonumber
	\end{align}
	where $\mathcal{A}_{10{\rm dB}} (\theta_{{\rm B},1},r_{{\rm B},1})$ is the 10-dB energy-spread angular support with respect to (w.r.t.)  near-field Bob.\footnote{The 10-dB energy-spread region is employed here to guarantee that the correlation outside the angular support $\mathcal{A}_{10{\rm dB}} (\theta_{{\rm B},1},r_{{\rm B},1})$ approximates 0.}

	The optimal transmit power allocation $P_{{\rm B},1}^{*} $ and $P_{{\rm B},2}^{*}$ for maximizing the achievable sum-rate can be obtained via the exhaustive search method. Moreover, the achievable sum-rate of the two Bobs can be upper-bounded by $$R_{{\rm B},1}+R_{{\rm B},2}\!\le\!\log_2\!\left(\! 1\!+\!\frac{P_{{\rm B},1} N g_{{\rm B},1}^2 }{  \sigma_{{\rm B}}^2}\! \right) +
			\log_2\!\left(\!1\!+\!\frac{P_{{\rm B},2} N g_{{\rm B},2}^2}{  \sigma_{{\rm B}}^2}\!\right), 
    $$ where the equality holds when $\chi_{{{\rm B}_1},{{\rm B}_2} }=0$. 

	\begin{figure}[t]
		\centering
		\includegraphics[width=0.35\textwidth]{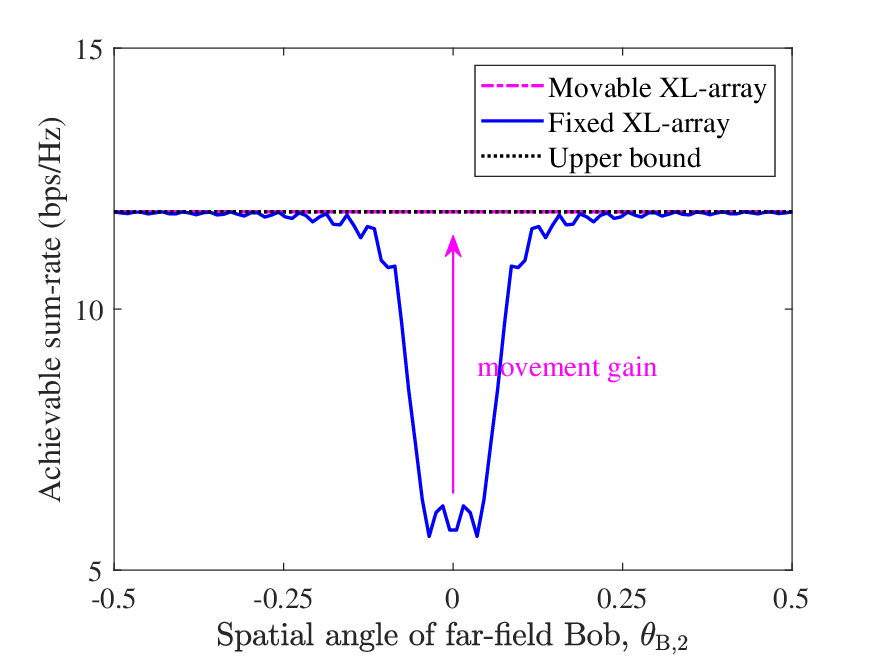}
		\caption{Achievable sum-rate versus spatial angle $\theta_{{\rm B},2}$ when $\theta_{{\rm B},1} = 0,~r_{{\rm B},1} = 10 $ m, and $ r_{{\rm B},2}=150 $ m.} \label{Fig:Sec3-MArate1}
		\vspace{-14pt}
	\end{figure}
	\begin{example}[Energy-spread effect]
		\rm In Fig.~\ref{Fig:Sec3-MArate1}, we plot the achievable sum-rate and its upper bound versus the spatial angle $\theta_{{\rm B},2}$, respectively. It is observed that the achievable sum-rate rapidly reduces in a specific angular region ($\theta_{{\rm B},2} \in \mathcal{A}_{10{\rm dB}} (\theta_{{\rm B},1},r_{{\rm B},1})$) because the orthogonality of channel steering vectors does not hold due to the energy-spread effect. This thus introduces significant inter-user interference between near-field Bob and far-field Bob, resulting in reduced achievable sum-rate. 
	\end{example}
	\subsubsection{\textbf{\underline{Movable XL-array}}}
	Next, we investigate the rate performance gain provided by movable XL-array. Specifically, the correlation function defined in~\eqref{Exp:Corfun} w.r.t. the positions of subarrays $\{{q}_m\}$ (i.e., $\chi(\mathbf{q},\{\theta_i,r_i\},\{\theta_{j},r_{j}\})$) can be rewritten as follows.
	\begin{lemma}\label{Lemma:Modular_approx}
		\rm 
		For  movable XL-array systems, the correlation function between channel steering vectors $\mathbf{a}(\mathbf{q},\theta_{i},r_{i}) $ and $\mathbf{a}(\mathbf{q},\theta_{j},r_{j})$ in~\eqref{Exp:Corfun} can be rewritten as
		\begin{align}
			&\chi(\mathbf{q},\{\theta_i,r_i\},\{\theta_{j},r_{j}\})  \nonumber \\
			& =  \frac{1}{N}  \Big| \sum_{m=1}^{M} \sum_{\tilde{n}=1}^{\tilde{N}}   e^{\jmath\frac{2\pi}{\lambda}({q}_{m} + \frac{2\tilde{n}-\tilde{N}-1}{2} d) \Omega_{\theta} + \jmath\frac{2\pi}{\lambda}({q}_{m} + \frac{2\tilde{n}-\tilde{N}-1}{2} d)^2 \Omega_{r} } \Big| \nonumber \\
			& \overset{(a)}{\approx} \frac{1}{N}\Bigg| \sum_{m=1}^{M} {\underset{\phi_{m}}{{\underbrace{e^{\jmath\frac{2\pi}{\lambda} ({q}_{m} \Omega_{\theta} + {q}_{m}^2 \Omega_{r} )  } }}}}  
			\underset{A_{m}}{\underbrace{\frac{\sin(\frac{\tilde{N}\pi \Omega_{m}}{2})}{\sin(\frac{\pi \Omega_{m}}{2})}}}
			\Bigg|, \label{Exp:Modular_cor}
		\end{align}
		where $ \Omega_{\theta} \triangleq \theta_{j} - \theta_{i} $, $ \Omega_{r} \triangleq \frac{1-\theta_{i}^2}{2 r_{i}} -\frac{1-\theta_{j}^2}{2 r_{j}} $, $\Omega_{m} \triangleq 2{q}_{m} \Omega_{r} + \Omega_{\theta} $, and $(a)$ holds when $r_{i}$ and $r_{j}$ are larger than $\frac{2(\tilde{N}d)^2}{\lambda}$.
	\end{lemma}
	\begin{proof}
		Please refer to Appendix~\ref{App:Modular_approx}.
	\end{proof}
	
	For movable XL-array, as shown in~\eqref{Exp:Modular_cor}, the correlation function can be expressed as a function of the positions of subarrays. By adjusting the position of each subarray, we can adjust the phase term $\{\phi_{m}\}$  and amplitude term $\{A_{m}\}$ in~\eqref{Exp:Modular_cor} to control the value of the correlation function.
	This demonstrates the potential of movable XL-array in addressing the mix-field interference arising from the energy-spread effect, since the achievable rates $R_{{\rm B},1}$ and $R_{{\rm B},2}$ decrease with the increase of $\chi_{{{\rm B}_1},{{\rm B}_2}  }(\mathbf{q}) = \chi(\mathbf{q},\{\theta_{{\rm B},1},r_{{\rm B},1}\},\{\theta_{{\rm B},2},r_{{\rm B},2}\})$. 
    
	According to~\eqref{Exp:Rate},  maximizing the achievable sum-rate is equivalent to minimizing the correlation function $\chi_{{{\rm B}_1},{{\rm B}_2}  }(\mathbf{q})$ in~\eqref{Exp:Rate} by optimizing the positions of subarrays. This problem can be mathematically formulated as follows
	\begin{subequations}
		\begin{align}
			(\textbf{P1}):\;\min_{\mathbf{q}}&\quad\chi_{{{\rm B}_1},{{\rm B}_2}  }(\mathbf{q}) \nonumber \\ 
			{\rm {s.t.}}&\quad|{q}_{i}-{q}_{j}|\ge \tilde{N} d,  ~~\forall i \neq j\in \mathcal{M},\label{C:XL-mMA1}\\
			&\quad{q}_{m} \in \mathcal{C}_{\mathbf{q}},~\forall m\in \mathcal{M},\label{C:XL-mMA2}
		\end{align} 
	\end{subequations}
	where constraints~\eqref{C:XL-mMA1} and~\eqref{C:XL-mMA2} specify the allowable movement region and minimum spacing between adjacent subarrays with $\mathcal{C}_{\mathbf{q}} = [{q}_{\rm min}, {q}_{\rm max} ]$.
    
	Note that Problem (\textbf{P1}) is a non-convex optimization problem due to the complicated near-field steering vectors, and thus is challenging to solve optimally. To address this issue and obtain useful insights, we first solve a relaxed problem of (\textbf{P2}) as follows, where the movement region constraint in~\eqref{C:XL-mMA2} is not considered.
	\begin{align}
		(\textbf{P2}):\;\min_{\mathbf{q}}~\chi_{{{\rm B}_1},{{\rm B}_2}  }(\mathbf{q}) 
		\quad\quad {\rm {s.t.}}~\eqref{C:XL-mMA1}. \nonumber
	\end{align}  
	\begin{proposition}\label{Theorem:opt_pos}
		\rm Given $\Omega_{\theta} $ and $\Omega_{r}$, an  optimal solution to Problem (\textbf{P2}) is given by
		\begin{align}
			{q}_{m} \in \left\{ \frac{2k-\tilde{N}\Omega_{\theta}}{2 \tilde{N}\Omega_{r}} \,\middle|\, k \in \mathbb{Z},\forall i\in\mathbb{Z},k\neq i\tilde{N}
			\right\}.
		\end{align} 
	\end{proposition}
	\begin{proof}
		\rm Let $ \sin(\frac{\tilde{N}\pi \Omega_{m}}{2}) = 0 $ and $ \sin(\frac{\pi \Omega_{m}}{2}) \neq 0 $. The value of $\Omega_{m}$ should satisfy the following condition
		\begin{align}\label{Exp:optOmega_m}
			\Omega_{m} \in \left\{ \frac{2k}{N} \,\middle|\, k \in \mathbb{Z},\forall i\in\mathbb{Z},k\neq i\tilde{N}  \right\}.
		\end{align}
		By substituting $\Omega_{m} = 2{q}_{m} \Omega_{r} + \Omega_{\theta}$ into~\eqref{Exp:optOmega_m}, the proof is thereby completed.
	\end{proof}
	
	On the other hand, for the original Problem (\textbf{P1}) with movement region constraint, the optimal solution generally may not be easily obtained in closed form, while its high-quality suboptimal solution can be obtained by using efficient algorithms, such as DE~\cite{das2010differential}; the details will be provided in Section~\ref{Sec:IV}. 
	\begin{remark}[Mixed-field interference nulling]
		\rm 
		To demonstrate the superiority of proposed movable XL-array scheme, we present in Fig.~\ref{Fig:Sec3-MArate1} the achievable sum-rate of the movable XL-array as well.
        It is observed that by adjusting the positions of subarrays, the proposed movable XL-array achieves close rate performance to its upper bound, indicating that interference nulling (or channel decorrelation) is almost achieved between mixed-field users, even under the array aperture constraint.
	\end{remark}

    \vspace{-18pt}
	\subsection{Covertness Analysis}\label{Sec:III-B}
	In this subsection, we analyze the covert transmission performance under mixed-field channel setup. Since the minimum DEP for each Willie, as indicated in~\eqref{Exp:MDEP}, is independent of the others, we thus consider one Willie (e.g., Willie 1) in the following performance analysis, while the results can be easily extended to the multi-Willie case.
	
	\subsubsection{\textbf{\underline{Fixed XL-array}}}
	For fixed XL-array systems, with the hybrid beamforming matrices provided in Section~\ref{Sec:III-A}, the average received power of information-bearing signals $f_{{\rm Co},1}$ at Willie is given by
	\begin{align}
		f_{{\rm Co},1}  = \sum_{b=1}^{2} P_{{\rm B},b} (\chi_{{\rm W}_{1},{\rm B}_{b} })^2 N  g_{{\rm W},1}^2,
	\end{align} 
	where $\chi_{{\rm W}_{1},{\rm B}_{b} }$ is the correlation between the channel steering vectors $ \mathbf{a}({\mathbf{q}_{\rm U},\theta_{{\rm W},1},r_{{\rm W},1}}) $ and $ \mathbf{a}({\mathbf{q}_{\rm U},\theta_{{\rm B},b},r_{{\rm B},b}}) $  with $ b\in\{1,2\}$. Given a positive constant $\varsigma$, the following lemma provides a condition for ensuring covert transmission. 
	\begin{lemma}[Covert transmission condition]\label{Lemma:CT_condition}
		\rm 
		For fixed XL-array systems, given the location of Willie 1 and $\epsilon = 1- \varsigma$, the covert transmission is guaranteed when
		\begin{align}
			\sum_{{b}=1}^{2} P_{{\rm B},b} (\chi_{{\rm W}_{1},{\rm B}_{b} })^2 \le  \frac{\tilde{\sigma}_{\rm w}^2}{N g_{{\rm W},1}^2} \! \left( \frac{\rho^{2(1-\epsilon)} - 1}{\rho}\right)\!  \triangleq \Gamma_{{\rm W}_{1}}(\epsilon). \nonumber
		\end{align}
	\end{lemma}
	\begin{proof}
		The proof of Lemma~\ref{Lemma:CT_condition} can be easily obtained from~\eqref{Exp:CT_constraint} with details omitted for brevity. 
	\end{proof}
	
	\begin{remark}[What affects covert transmission condition?]
		\rm 
		{\bf Lemma~\ref{Lemma:CT_condition}}  reveals that several factors influence the covert transmission condition. 
		1) Complex-valued path gain $g_{{\rm W},1}$: a weaker path-gain results in a higher $\Gamma_{{\rm W}_{1}}(\epsilon)$, leading to a more relaxed condition for covert transmission;
		2) DEP requirement $\epsilon$: a higher $\epsilon$ results in a  stricter condition for covert transmission; 
		3) Transmit power $P_{{\rm B},b}, b\in\{1,2\}$: increasing transmit power results in a more stringent condition for covert transmission; 
		and 4) Channel correlations $\chi_{{\rm W}_{1},{\rm B}_{b} }, b\in\{1,2\}$: a higher correlation $\chi_{{\rm W}_{1},{\rm B}_{b} } $ may lead to the failure of covert transmission condition.
	\end{remark}

	Among these factors, the channel correlation $ \chi_{{\rm W}_{1},{\rm B}_{b} }, b\in\{1,2\} $ is particularly critical, which directly affects the covert transmission condition when the location of Willie, DEP requirement~$\epsilon$, and transmit power $P_{{\rm B},b}, b\in\{1,2\}$ are given. To facilitate analysis, we first divide the mixed-field covert communication scenario into the following four cases, based on the values of $ \chi_{{\rm W}_{1},{\rm B}_{b} }, b\in\{1,2\} $. 
    \begin{itemize}
    \item {\bf Case 1} (No power leakage): The channel steering vectors between both Bobs and Willie are nearly orthogonal (i.e., $\chi_{{\rm W}_{1},{\rm B}_{1} } \approx 0$ and $\chi_{{\rm W}_{1},{\rm B}_{2} } \approx 0$), which is mathematically represented as
		\begin{align}
			\mathbf{C}_{1} \triangleq 
			\big\{\theta_{{\rm B},1} \neq \theta_{{\rm W},1},
			\theta_{{\rm B},2} \notin {\mathcal{A}}_{\mu{\rm dB}} (\theta_{{\rm W},1},r_{{\rm W},1})
			\big\}. \nonumber
		\end{align}
      
	\item {\bf Case 2} (Power leakage from far-field Bob): Only the correlation between far-field Bob and Willie  is non-negligible (i.e., $\chi_{{\rm W}_{1},{\rm B}_{1} } \approx 0$ and $\chi_{{\rm W}_{1},{\rm B}_{2} } \neq 0$), which is  expressed~as
		\begin{align}
			\mathbf{C}_{2} \triangleq 
			\big\{\theta_{{\rm B},1} \neq \theta_{{\rm W},1},
			\theta_{{\rm B},2} \in {\mathcal{A}}_{\mu{\rm dB}} (\theta_{{\rm W},1},r_{{\rm W},1})
			\big\}. \nonumber
		\end{align}
        
	\item {\bf Case 3} (Power leakage from near-field Bob): Only the correlation between near-field Bob and Willie is non-negligible (i.e., $\chi_{{\rm W}_{1},{\rm B}_{1} } \neq  0$ and $\chi_{{\rm W}_{1},{\rm B}_{2} } \approx 0$), which is  expressed as
		\begin{align}
			\mathbf{C}_{3} \triangleq 
			\big\{\theta_{{\rm B},1} = \theta_{{\rm W},1},
			\theta_{{\rm B},2} \notin {\mathcal{A}}_{\mu{\rm dB}} (\theta_{{\rm W},1},r_{{\rm W},1})
			\big\}. \nonumber
		\end{align}
	\item {\bf Case 4} (Power leakage from both Bobs):
	For both Bobs, the correlations between their channel steering vectors and those of Willie are non-negligible
		(i.e., $\chi_{{\rm W}_{1},{\rm B}_{1} } \neq 0$ and $\chi_{{\rm W}_{1},{\rm B}_{2} } \neq 0$), which is  represented by
		\begin{align}
			\mathbf{C}_{4} \triangleq 
			\big\{\theta_{{\rm B},1} = \theta_{{\rm W},1},
			\theta_{{\rm B},2} \in {\mathcal{A}}_{\mu{\rm dB}} (\theta_{{\rm W},1},r_{{\rm W},1})
			\big\}. \nonumber
		\end{align}
    \end{itemize}

 For Case 1, the conditions for covert transmission can always be satisfied; thus we focus on the other three cases to draw useful insights. Specifically, for Case 2, the power leakage caused by energy-spread results in a significantly large value of $\chi_{{\rm W}_{1},{\rm B}_{2} }$. In the following, we first present the covert transmission condition for Case 2. 
	\begin{lemma}[Covert transmission condition of far-field Bob]\label{lem:CT_conB2}
		\rm 
		Given the location of near-field Willie,
		when $\theta_{{\rm B},1}\!\neq\!\theta_{{\rm W},1}$
		and $\theta_{{\rm B},2} \in {\mathcal{A}}_{\mu{\rm dB}}   (\theta_{{\rm W},1},r_{{\rm W},1})$, the covert transmission condition in \textbf{Lemma~\ref{Lemma:CT_condition}} reduces to 
		\begin{align}
			\chi_{{\rm W}_{1},{\rm B}_{2} } \le  \sqrt{  {\Gamma_{{\rm W}_{1}}(\epsilon)}/{P_{{\rm B},2}}}  \triangleq {\Delta}_{2}(\epsilon).
		\end{align}
	\end{lemma}
	\begin{proof}
		The proof  can be easily obtained based on Lemma~\ref{Lemma:CT_condition}, and thus is omitted here for brevity. 
	\end{proof}
	
	\textbf{Lemma~\ref{lem:CT_conB2}} indicates that the channel correlation between near-field Willie and far-field Bob caused by the energy-spread effect should be sufficiently small to guarantee covertness.  In addition, an increasing $\epsilon$ yields a smaller ${\Delta}_{2}(\epsilon)$ (i.e., a more stringent covert transmission condition). 
    Moreover, as $ \mathbf{a}({\mathbf{q}_{\rm U},\theta_{{\rm B},2},r_{{\rm B},2}}) $ can be approximated by a far-field steering vector $\hat{\mathbf{a}}({\mathbf{q}_{\rm U},\theta_{{\rm B},2}})$,  the covert transmission condition for far-field Bob depends on its spatial angle.  Based on the above, we characterize the covert transmission angle for far-field Bob, for which its transmission covertness can be guaranteed. 
	\begin{proposition}[Covert transmission angle of far-field Bob]\label{Pro:CTA}
		\rm 
		Given the location of near-field Willie, when  $\theta_{{\rm B},1} \neq \theta_{{\rm W},1}$ and
		$\theta_{{\rm B},2} \in {\mathcal{A}}_{\mu{\rm dB}} (\theta_{{\rm W},1},r_{{\rm W},1})$, the covert transmission angle of far-field Bob is obtained as
		\begin{align}
        \label{eq:covert region}
			&\Big\{ \theta_{{\rm B},2} \le \Xi_{\rm left}\left(
            \theta_{{\rm W},1},r_{{\rm W},1},{\Delta}_{2}(\epsilon)
            \right)\Big\} \nonumber \\
			&\quad \quad  \cup
			\Big\{	\theta_{{\rm B},2} \ge \Xi_{\rm right}\left(
                \theta_{{\rm W},1},r_{{\rm W},1},{\Delta}_{2}(\epsilon)
                \right) \Big\}.
		\end{align}
	Herein, {\small$\Xi_{\rm left}(\theta_{{\rm W},1},r_{{\rm W},1},{\Delta}_{2}(\epsilon)) $} and  {\small$\Xi_{\rm right}(\theta_{{\rm W},1},r_{{\rm W},1},{\Delta}_{2}(\epsilon)) $} respectively denote the two solutions to the equation {$f\left( \mathbf{w}\left(\theta \right),\theta_{{\rm W},1},r_{{\rm W},1}  \right) = {\Delta}_{2}(\epsilon)$}.
	\end{proposition}
	\begin{proof}
        Given a threshold ${\Delta}_{2}(\epsilon)$, if $\theta_{{\rm B},2} \in {\mathcal{A}}_{\mu{\rm dB}} (\theta_{{\rm W},1},r_{{\rm W},1})$, the equation $f\left( \mathbf{w}\left(\theta \right),\theta_{{\rm W},1},r_{{\rm W},1}  \right) = {\Delta}_{2}(\epsilon)$ mainly has two solutions, denoted by {$\Xi_{\rm left}(\theta_{{\rm W},1},r_{{\rm W},1},{\Delta}_{2}(\epsilon)) $} and  {$\Xi_{\rm right}(\theta_{{\rm W},1},r_{{\rm W},1},{\Delta}_{2}(\epsilon)) $}  (as illustrated in Fig. \ref{Fig:CTangle}).\footnote{Although the equation $f\left( \mathbf{w}\left(\theta\right),\theta_{{\rm W},1},r_{{\rm W},1}  \right) = {\Delta}_{2}(\epsilon)$ may have multiple solutions due to power fluctuation, these solutions are very close to each other. For convenience, we focus only on the smallest and largest solutions~\cite{liuty2025PLS}.}
        As such, the covert transmission angle of far-field Bob can be directly obtained as that in \eqref{eq:covert region}.
	\end{proof}
	
	\textbf{Proposition~\ref{Pro:CTA}} indicates that to ensure the covert transmission condition, the spatial angle difference between Willie and far-field Bob must exceed a certain threshold, which is determined by  ${\Delta}_{2}(\epsilon)$. 
    \begin{figure}[t]
		\centering
		\includegraphics[width=0.3\textwidth]{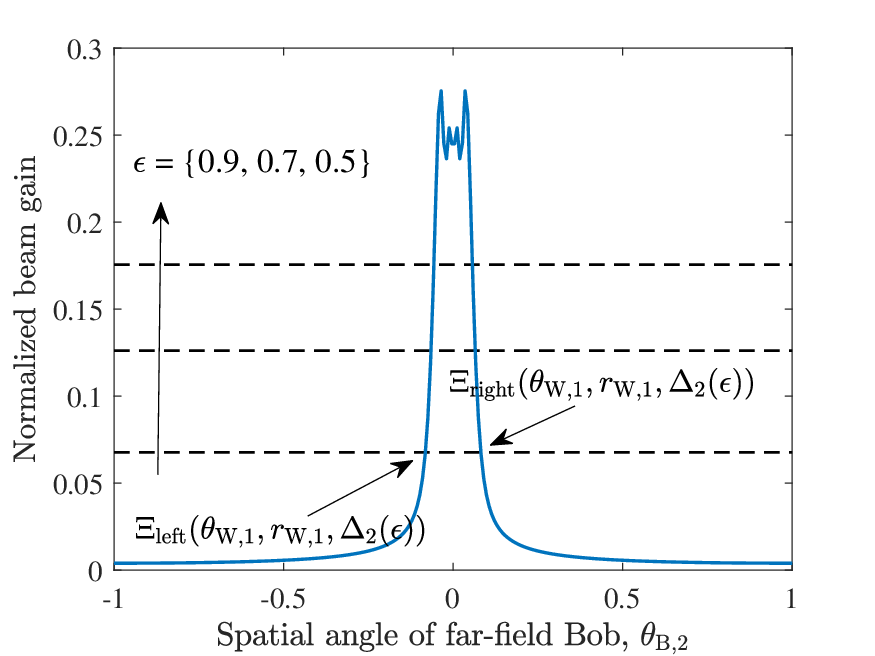}
		\caption{Covert transmission angle of far-field Bob with   $\theta_{{\rm W},1}= 0$ and $ r_{{\rm W},1} = 10$ m for fixed XL-array scheme.} \label{Fig:CTangle}
		\vspace{-16pt}
   \end{figure}
	\begin{example}
		\rm 
		In Fig.~\ref{Fig:CTangle}, we plot the covert transmission angle of far-field Bob when $\theta_{{\rm W},1}= 0$ and $ r_{{\rm W},1} = 10$ m for fixed XL-array systems. Herein, the smallest and largest spatial angles $\Xi_{\rm left}$ and $\Xi_{\rm right}$ w.r.t. $\{\theta_{{\rm W},1},r_{{\rm W},1},{\Delta}_{2}(\epsilon) \}$ can be obtained based on \textbf{Proposition~\ref{Pro:CTA}}. To satisfy the covert transmission condition for far-field Bob, the allowable spatial angle should be smaller than $\Xi_{\rm left}$ or larger than $\Xi_{\rm right}$ to ensure the channel correlation $\chi_{{\rm W}_{1},{\rm B}_{2} }$ lower than a given threshold (i.e., black dashed lines in Fig.~\ref{Fig:CTangle}). In addition, given a larger  $\epsilon$ (more stringent covert transmission condition), a larger spatial angle difference (a smaller $\Xi_{\rm left}$ and/or a larger $\Xi_{\rm right}$) is needed to ensure covertness.
	\end{example}
		
	For Case 3, when the near-field Bob locates at the same angle with near-field Willie, their channel correlation is affected by their ranges, which is characterized below.
	\begin{lemma}[Covert transmission condition of near-field Bob]\label{Lem:CTRB1}
		\rm 
		Given the location of near-field Willie, when $\theta_{{\rm B},1} = \theta_{{\rm W},1}$ and
		$\theta_{{\rm B},2} \notin {\mathcal{A}}_{\mu{\rm dB}} (\theta_{{\rm W},1},r_{{\rm W},1})$,
        the covert transmission condition of near-field Bob is given by
		\begin{align}\label{Exp:CorW1B1_condition}
			\chi_{{\rm W}_{1},{\rm B}_{1} } \le \sqrt{{\Gamma_{{\rm W}_{1}}(\epsilon)}/{P_{{\rm B},1}}  } \triangleq {\Delta}_{1}(\epsilon).
		\end{align}
	\end{lemma}
	\begin{proof}
		The proof can be easily obtained based on Lemma~\ref{Lemma:CT_condition} and thus is omitted here for brevity. 
	\end{proof}

    Based on \textbf{Lemma~\ref{Lem:CTRB1}}, we obtain the range region of near-field Bob, where its transmission covertness can be guaranteed, namely, the covert transmission range.
	\begin{proposition}[Covert transmission range of near-field Bob]\label{Pro:CTR}
		\rm Let $ \beta_{\Delta} $ denote a nonnegative constant that satisfies $|G(\beta_{\Delta})| = \Delta $.
		Given the location of Willie  $\{\theta_{{\rm W},1},r_{{\rm W},1}\}$, when $\theta_{{\rm W},1} =\theta_{{\rm B},1}  $, the covert transmission range of Bob $1$ is given by
		\begin{align}
			\Big\{ r_{{\rm B},1} \le r_{{\rm W},1}\Big(\frac{1}{1 + \Pi r_{{\rm W},1}}\Big)\Big\} \cup
			\Big\{	r_{{\rm B},1} \ge r_{{\rm W},1}\Big(\frac{1}{1- \Pi r_{{\rm W},1}}\Big)  \Big\},
			\nonumber
		\end{align}
		where $\Pi = \frac{2\lambda (\beta_{{\Delta}_{1}(\epsilon)})^2}{N^2d^2 (1-\theta_{{\rm W},1}^2)}$. 
	\end{proposition}
	\begin{proof}
		Please refer to  Appendix~\ref{App:CT_region}
	\end{proof}
    
	In \textbf{Proposition~\ref{Pro:CTR}}, we have $ \frac{1}{1 + \Pi r_{{\rm W},1}} < 1 $ and $\frac{1}{1- \Pi r_{{\rm W},1}} > 1 $. This indicates that for ensuring covert transmission, the range difference between near-field Bob and Willie must exceed a certain threshold, the value of which is related to $\Pi$. Additionally, as $\beta_{{\Delta}_{1}(\epsilon)}$ increases with a decreasing ${\Delta}_{1}(\epsilon)$, a more stringent condition for covert transmission (i.e., lower ${\Delta}_{1}(\epsilon)$) in~\eqref{Exp:CorW1B1_condition} results in a larger $\Pi$, indicating that a greater range difference between near-field Bob and Willie is necessary.

	
	 For the more challenging Case 4,  the sum-correlation (i.e., $\sum_{b=1}^{2}P_{{\rm B},b}(\chi_{{\rm W}_{1},{\rm B}_{b} })^2$) between Willie and Bobs must be kept below a threshold corresponding to $\epsilon$ (i.e., $\Gamma_{{\rm W}_{1}}(\epsilon)$) for ensuring covert transmission, as mentioned in \textbf{Lemma~\ref{Lemma:CT_condition}}. Specifically, given the locations of near-field Bob and Willie, the covert transmission condition of far-field Bob is given by
     \begin{align}
	 \chi_{{\rm W}_{1},{\rm B}_{2} } \le  \sqrt{ \frac{\Gamma_{{\rm W}_{1}}(\epsilon) - P_{{\rm B},1} (\chi_{{\rm W}_{1},{\rm B}_{1} })^2  }{P_{{\rm B},2}}}  \triangleq \tilde{\Delta}_{2}(\epsilon).
     \end{align}
     By replacing ${\Delta}_{2}(\epsilon)$ in \textbf{Proposition~\ref{Pro:CTA}} with $\tilde{\Delta}_{2}(\epsilon)$,  we can obtain the covert transmission angle of far-field Bob for Case 4 as well. Since $\chi_{{\rm W}_{1},{\rm B}_{1} } > 0$, we always have  $\tilde{\Delta}_{2}(\epsilon) \le {\Delta}_{2}(\epsilon)$ for this case.
     This indicates that a larger angle difference is needed between Willie and far-field Bob to ensure covertness. 
     
	\subsubsection{\textbf{\underline{Movable XL-array}}}
	For movable XL-array systems, the channel correlation can be reduced by flexibly adjusting the positions of  subarrays. To this end, we formulate an optimization problem to minimize the average received power of information-bearing signals $f_{{\rm Co},1}$ for reducing the power leakage from both near-field and far-field Bobs.  Accordingly, the equivalent optimization problem is formulated as follows 
		\begin{align}
			(\textbf{P3}):\;\min_{\mathbf{q}}~ \sum_{b=1}^{2}P_{{\rm B},b}\big(\chi_{{\rm W}_{1},{\rm B}_{b} } (\mathbf{q})\big)^2 
			\quad {\rm {s.t.}}~\eqref{C:XL-mMA1},~\eqref{C:XL-mMA2}. \nonumber
		\end{align} 
	Problem (\textbf{P3}) is a non-convex optimization problem due to the non-convex expression of $\chi_{{\rm W}_{1},{\rm B}_{b} } (\mathbf{q})$.
	To tackle this difficulty, efficient algorithms (e.g., DE) can be applied to solve Problem (\textbf{P3}). As the closed-form solution to Problem (\textbf{P3}) is intractable, we provide an example below to demonstrate the effectiveness of movable XL-array in reducing the channel correlations between Bobs and Willie to enhance covertness.
	\begin{figure}[t]
		\centering
		\includegraphics[width=0.3\textwidth]{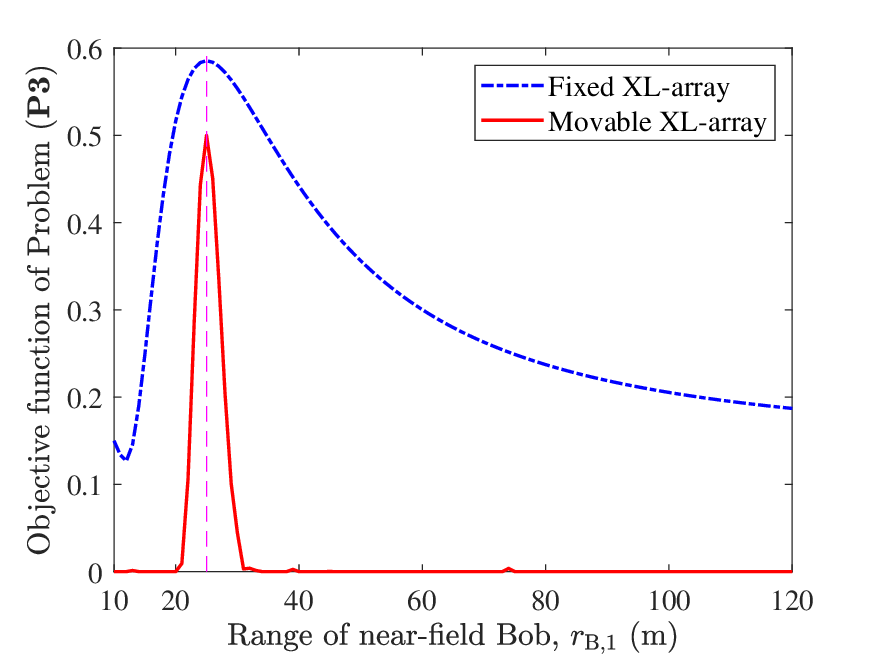}
		\caption{Objective function of Problem \textbf{(P3)} with
			$\theta_{{\rm B},1}= \theta_{{\rm B},2}=\theta_{{\rm W},1} = 0 $, $r_{{\rm W},1} = 25$ m, and $r_{{\rm B},2} = 150$ m.} \label{Fig:Sec3-Obj3}
		\vspace{-16pt}
	\end{figure}
	
	\begin{example}
		\rm 
		In Fig.~\ref{Fig:Sec3-Obj3}, we show the objective function of Problem \textbf{(P3)} versus the range of near-field Bob, when $\theta_{{\rm B},1}= \theta_{{\rm B},2}=\theta_{{\rm W},1} = 0 $, $r_{{\rm W},1} = 25$, and $r_{{\rm B},2} = 150$ m. It is observed that, for fixed XL-array, the objective function of Problem \textbf{(P3)} is larger than $0.1$, due to the strong correlation between Willie and Bobs. 
		In contrast, for  movable XL-array, the objective function of Problem \textbf{(P3)} significantly reduces and is close to zero. However, when the range of Willie  is between $20$ m and $30$ m, the $\sum_{b=1}^{2}P_{{\rm B},b}\big(\chi_{{\rm W}_{1},{\rm B}_{b} } (\mathbf{q})\big)^2 $ is relatively high. This is due to the significantly strong channel correlation between near-field Willie and Bob, which is difficult to address under the array aperture constraint.
	\end{example}
	\vspace{-12pt}
	\subsection{Trade-off between Rate and Covertness Performance}
	The performance analysis in Sections~\ref{Sec:III-A} and~\ref{Sec:III-B} show that movable XL-array can be employed to improve either the achievable sum-rate or to enhance transmission covertness. However, for covert communication in mixed-field systems, how to improve both communication performance and transmission covertness is more challenging, which needs to balance trade-off between them. 
	\begin{itemize}
		\item {\bf(Subarray movement)} For sum-rate performance, the subarray movement should be designed to minimize the channel correlations among Bobs. Nevertheless, for enhancing transmission covertness, the main principle of subarray movement adjustment is to minimize the channel correlations between Bobs and Willies.
	\item {\bf(Power allocation)} To maximize the achievable sum-rate, all available transmit power should be utilized. In contrast, to ensure covertness, using smaller transmit power only can reduce the risk of being detected.	
	\end{itemize}
	
	\section{Problem Formulation and Proposed Solution}\label{Sec:IV}
	To investigate the performance of MA-enhanced mixed-field covert communication systems, in this section, we formulate an optimization problem to maximize the achievable sum-rate under constraints on covertness, hybrid beamforming, transmit power, and subarray movement region. An efficient algorithm is then proposed to solve this non-convex problem.
	
	\vspace{-6pt}
	\subsection{Problem Formulation}
	We aim to maximize the covert transmission rate by jointly designing the hybrid beamforming matrices $\{\mathbf{F}_{\rm A},\mathbf{F}_{\rm D}\}$ and the positions of movable subarrays $\mathbf{q}$. This optimization problem can be formulated~as\footnote{In this work, we focus on maximizing the achievable sum-rate without enforcing per-user QoS constraints, while the proposed algorithm can be extended to scenarios where such per-user QoS requirements are considered.}
	\begin{subequations}
		\begin{align}
			(\textbf{P4}):\;\max_{\mathbf{q}, \mathbf{F}_{\rm A}, \mathbf{F}_{\rm D}}&~ \sum_{b=1}^{B} \varpi_{b} R_{{\rm B},b}(\mathbf{q},\mathbf{F}_{\rm A},\mathbf{F}_{\rm D}) \nonumber  \\
			{\rm {s.t.}}&~ f_{{\rm Co},w}(\mathbf{q},\mathbf{F}_{\rm A},\mathbf{F}_{\rm D}) \le \tilde{\sigma}_{\rm w}^2 \left(   \frac{\rho^{2\varsigma} - 1}{\rho}\right), \forall w \in \mathcal{W},  \label{C:Covert} \\
			&~ |  [\mathbf{F}_{\rm A}]_{n,b} |=1, ~~~\forall n\in \mathcal{N},~b\in \mathcal{B},\label{C:Phase} \\
			&~ || \mathbf{F}_{\rm A} \mathbf{F}_{\rm D}  ||_F^2 \le P, \label{C:Power} \\
			&~ \eqref{C:XL-mMA1},~\eqref{C:XL-mMA2}, \nonumber
		\end{align} 
	\end{subequations}
	where $\varpi_{b}$ is the weighting factor. Herein, \eqref{C:Covert} ensures covert transmissions between Alice and Bobs, \eqref{C:Phase} is the unit-modulus constraint for analog beamforming, and \eqref{C:Power} limits the maximum transmit power. 
	
    Problem (\textbf{P4}) is a non-convex optimization problem, which is generally difficult to solve. In particular, the conventional weighted sum-rate maximizing methods (e.g., weighted MMSE and fractional programming (FP) methods) cannot be directly applied to solve this problem due to the non-convex covertness constraints for each Willie, as well as the complicated constraints related to subarray movement $\mathbf{q}$.

	\vspace{-10pt}
	\subsection{Proposed Algorithm to Solve \textbf{(P4)}}
	To solve Problem (\textbf{P4}), we first transform it into two subproblems below, corresponding to an inner problem for optimizing the hybrid beamforming matrices $\{\mathbf{F}_{\rm A},\mathbf{F}_{\rm D}\}$ given positions of subarrays, and an outer problem for further optimizing the positions of subarrays $\mathbf{q}$.
	
	\underline{\textbf{Inner problem}}: Given any feasible positions of subarrays $\mathbf{q}$, the inner problem of beamforming optimization is given by  
	\begin{align}
		(\textbf{P5}):\; \max_{\mathbf{F}_{\rm A},\mathbf{F}_{\rm D}}&\quad \sum_{b=1}^{B} \varpi_{b} R_{{\rm B},b}(\mathbf{q},\mathbf{F}_{\rm A},\mathbf{F}_{\rm D}) \nonumber  \\
		{\rm {s.t.}}&\quad \eqref{C:Covert},~\eqref{C:Phase},~\eqref{C:Power}.  \nonumber
	\end{align}
    
	\underline{\textbf{Outer problem}}: 
	Let $\hat{\mathbf{F}}_{\rm A}(\mathbf{q})$ and $\hat{\mathbf{F}}_{\rm D}(\mathbf{q})$ denote the obtained solution to Problem (\textbf{P5}). The outer problem of subarray-position optimization is formulated as  
	\begin{align}
		(\textbf{P6}):\; \max_{\mathbf{q}}&\quad \sum_{b=1}^{B} \varpi_{b} R_{{\rm B},b}(\mathbf{q},\hat{\mathbf{F}}_{\rm A}(\mathbf{q}),\hat{\mathbf{F}}_{\rm D}(\mathbf{q})) \nonumber  \\
		{\rm {s.t.}}&\quad \eqref{C:XL-mMA1},~\eqref{C:XL-mMA2}.  \nonumber
	\end{align}

	\subsubsection{Proposed solution to Problem (\textbf{P5})} Given the positions of subarrays $\{\mathbf{q}\}$, Problem (\textbf{P5}) is still hard to solve directly due to the covert transmission constraint~\eqref{C:Covert} and unit-modulus constraint~\eqref{C:Phase}. In addition, the coupling of hybrid beamforming matrices in both the objective function and constraints of (\textbf{P5}) makes the problem even more challenging. 
	
	To address these issues, we adopt an efficient and low-complexity two-stage beamforming design framework to obtain a high-quality solution~\cite{LiuTSBeam,Wangzl}. 

	\textbf{Stage 1} (Analog beamforming):
	First, the analog beamforming matrix $\mathbf{F}_{\rm A}$ is designed to maximize the received signal powers at individual Bobs. Specifically, based on the LoS path of each Bob, the $b$-th column of $\hat{\mathbf{F}}_{\rm A} (\mathbf{q})$, denoted by $\hat{\mathbf{f}}_{{\rm A},b}(\mathbf{q})$, can be obtained as\footnote{In high-frequency band scenarios, the LoS channel is much more dominant than the NLoS components. Therefore, in this work, the analog beamforming design is based on the LoS channel, which has been shown to achieve performance comparable to that of the fully digital approximation (FDA)-based hybrid beamforming scheme~\cite{Wangzl}.} 
	\begin{align}\label{Exp:ana_des}
		\hat{\mathbf{f}}_{{\rm A},b}(\mathbf{q}) =\sqrt{N} \mathbf{a}(\mathbf{q},\theta_{{\rm B},b}, r_{{\rm B},b}), \forall b \in \mathcal{B}.
	\end{align}
	
	\textbf{Stage 2} (Digital beamforming design): 
    Given the analog beamforming  $\hat{\mathbf{F}}_{\rm A} $, we then optimize the digital beamforming matrix. Specifically, with the obtained $\hat{\mathbf{F}}_{\rm A} $, the effective channels for all Bobs and Willies are given by
	\begin{subequations}
		\begin{align}
			\mathbf{e}_{{\rm B},b}(\mathbf{q}) &= \hat{\mathbf{F}}^H_{\rm A}(\mathbf{q}) \mathbf{h}_{{\rm B},b}(\mathbf{q}),~\forall b \in \mathcal{B}, \\
			\mathbf{e}_{{\rm W},w}(\mathbf{q}) &= \hat{\mathbf{F}}^H_{\rm A}(\mathbf{q}) \mathbf{h}_{{\rm W},w}(\mathbf{q}),~\forall w \in \mathcal{W}.
		\end{align}
	\end{subequations}
	With effective channels, the achievable rate at Bob $b$ and the covert transmission constraint for Willie $w$ can be rewritten as
	\begin{subequations}
		\begin{align}
			&\hat{R}_{b}(\mathbf{q},\mathbf{F_{\rm D}}) = \log_2\left(  1 + \frac{ |\mathbf{e}^H_{{\rm B},b}(\mathbf{q})  \mathbf{f}_{{\rm D},b}   |^2}{ \sum_{i=1,i\neq b}^{B} |\mathbf{e}^H_{{\rm B},b}(\mathbf{q})  \mathbf{f}_{{\rm D},i}   |^2 + \sigma_{{\rm B}}^2  } \right), \nonumber \\
			&~~~~~~~~~~~~~~~~~~~~~~~~~~~~~~~~~~~~~~~~~~~~~~~~~~~~~\forall b \in \mathcal{B}, \\
			&\hat{f}_{{\rm Co},w}(\mathbf{q},\mathbf{F_{\rm D}})   = \sum_{b=1}^{B}|\mathbf{e}^H_{{\rm W},w}(\mathbf{q}) \mathbf{f}_{{\rm D},b} |^2 \le \tilde{\sigma}_{\rm w}^2 \left(   \frac{\rho^{2\varsigma} - 1}{\rho}\right) , \nonumber \\
			&~~~~~~~~~~~~~~~~~~~~~~~~~~~~~~~~~~~~~~~~~~~~~~~~~~~\forall w \in \mathcal{W}, \label{C:Covert_new}
		\end{align}
	\end{subequations}
	respectively. Based on the above, the inner problem of beamforming optimization reduces to the following problem for the digital beamforming optimization
	\begin{align}
		(\textbf{P7}):\; \max_{\mathbf{F}_{\rm D}}
		&\quad \sum_{b=1}^{B} \varpi_{b} \hat{R}_{{\rm B},b}(\mathbf{q},\mathbf{F}_{\rm D}) \nonumber  \\
		{\rm {s.t.}}&\quad ~\eqref{C:Power},~\eqref{C:Covert_new}. \nonumber 
	\end{align}
Since the objective function in Problem (\textbf{P7}) is non-convex w.r.t. $\mathbf{F}_{\rm D}$, it is still challenging to optimize directly. To tackle this difficulty, we employ the SCA technique~\cite{MyTPM} to obtain a concave surrogate function for approximating the non-convex terms, i.e., $\hat{R}_{{\rm B},b}$ in Problem (\textbf{P7}). Specifically, the surrogate function of $\hat{R}_{{\rm B},b}$, denoted by $\tilde{R}_{{\rm B},b}$, is given by
	\begin{align}
		\tilde{R}_{{\rm B},b} & \triangleq \frac{1}{\ln2} \Bigg(  
		\ln  \bigg( 1+\frac{| {\zeta}_{b}^{\left(k \right) } |^2}{  {\vartheta}_{b}^{\left(k \right) }  } \bigg)     
		- \frac{| {\zeta}_{b}^{\left(k \right) } |^2}{  {\vartheta}_{b}^{\left(k \right) }  } 
		\nonumber \\
		& + \frac{2\mathcal{R}\left\lbrace (  {\zeta}_{b}^{\left(k \right) } )^H  {\zeta}_{b}  \right\rbrace }{{\vartheta}_{b}^{\left(k \right) }}
		-\frac{| {\zeta}_{b}^{\left(k \right) } |^2 \big(  | {\zeta}_{b} |^2  + {\vartheta}_{b}  \big)  }
		{ {\vartheta}_{b}^{\left(k \right) } \big(| {\zeta}_{b}^{\left(k \right) } |^2+   {\vartheta}_{b}^{\left(k \right) }\big) }
		\Bigg),
	\end{align}
	where ${\zeta}_{b} = \mathbf{e}^H_{{\rm B},b}(\mathbf{q})  \mathbf{f}_{{\rm D},b}$,  ${\vartheta}_{b} = \sum_{i=1,i\neq b}^{B} |\mathbf{e}^H_{{\rm B},b}(\mathbf{q})  \mathbf{f}_{{\rm D},i}   |^2 + \sigma_{{\rm B}}^2$, $ {\zeta}_{b}^{(k)} = \mathbf{e}^H_{{\rm B},b}(\mathbf{q})  \mathbf{f}_{{\rm D},b}^{(k)} $, and ${\vartheta}_{b}^{(k)} = \sum_{i=1,i\neq b}^{B} |\mathbf{e}^H_{{\rm B},b}(\mathbf{q})  \mathbf{f}_{{\rm D},i}^{(k)}   |^2 + \sigma_{{\rm B}}^2$. Herein, $ \mathbf{F}_{{\rm D}}^{(k)}  = [\mathbf{f}_{{\rm D},1}^{(k)},\ldots,\mathbf{f}_{{\rm D},B}^{(k)} ] $ is the digital beamforming matrix obtained in the $k$-th SCA iteration. By constructing a tight concave lower bound for $\hat{R}_{{\rm B},b}$ (i.e., $\hat{R}_{{\rm B},b} \ge \tilde{R}_{{\rm B},b}$), the original Problem (\textbf{P7}) can be approximated by the following convex optimization problem
	\begin{align}
		(\textbf{P8}):\; \max_{\mathbf{F}_{\rm D}}~\sum_{b=1}^{B} \varpi_{b} \tilde{R}_{{\rm B},b}(\mathbf{q},\mathbf{F}_{\rm D}) 
		\quad{\rm {s.t.}}~\eqref{C:Power},~\eqref{C:Covert_new}, \nonumber 
	\end{align}
	which can be efficiently solved by using e.g., CVX solvers. By iteratively updating $ \mathbf{F}_{{\rm D}}^{(k)} $ until convergence is achieved, we obtain a suboptimal digital beamforming matrix $\hat{\mathbf{F}}_{\rm D}(\mathbf{q})$. 
	
	\subsubsection{Proposed solution to Problem (\textbf{P6})}
	For the outer problem, there generally lacks a tractable closed-form expression for the optimized digital beamforming matrix  $\hat{\mathbf{F}}_{\rm D}$, which is obtained by solving the inner Problem \textbf{(P5)}. Moreover, the non-convex expressions of near-field steering vectors w.r.t. the position vector~$\mathbf{q}$ make the problem more difficult. 
	Although some existing methods, such as projected gradient ascent (PGA)~\cite{Wuqq_MA_secure}, can be employed to optimize the position vector, this approach is liable to converge to a local optimum or even a low-quality solution.  To address this issue, we propose a customized DE algorithm to obtain a high-quality solution to Problem \textbf{(P6)}. In the following, the framework of the proposed DE algorithm for the position optimization of movable XL-array is presented.
    
	\textbf{Population initialization}:
	At the beginning of DE algorithm, a population of $G$ feasible position vectors for the movable XL-array are randomly generated within the predefined solution space, which are expressed as
	\begin{align}\label{Exp:pop_gen}
		\mathcal{G}_{\rm q}^{(0)} = \Big\{\mathbf{q}_{1}^{(0)},\ldots,\mathbf{q}_{g}^{(0)},\ldots,\mathbf{q}_{G}^{(0)}\Big\}, \forall g\in \mathcal{G}, 
	\end{align}
	where $\mathbf{q}_{g}^{(0)}$ denotes the $g$-th candidate position vector of movable XL-array, $ \mathcal{G} = \{1,2,\ldots,G\}$ represents the population index set, and $G$ is the population size. Let $ \mathbf{q}_{g}^{(i)} $ denote the $g$-th position vector in the $i$-th DE iteration, where $i\in \mathcal{I} \triangleq \{1,2,\ldots,I\}$ with $I$ representing its maximum iteration number and $\mathcal{I}$ denoting its  index set.
	
	\textbf{Fitness function design}:
	For each candidate position vector $ \mathbf{q}_{g}^{(i)} $, the corresponding hybrid beamforming matrices $\big\{\hat{\mathbf{F}}_{\rm A}(\mathbf{q}_{g}^{(i)}),\hat{\mathbf{F}}_{\rm D}(\mathbf{q}_{g}^{(i)})\big\}$ can be obtained by solving Problem (\textbf{P5}). 
	Based on the above, the quality of  $\mathbf{q}_{g}^{(i)}$ is evaluated based on the following fitness function~\cite{das2010differential}
	\begin{align}\label{Exp:fit_fun}
		\mathcal{F}( \mathbf{q}_{g}^{(i)}) &= \sum_{b=1}^{B} R_{{\rm B},b}(\mathbf{q}_{g}^{(i)},\hat{\mathbf{F}}_{\rm A}(\mathbf{q}_{g}^{(i)}),\hat{\mathbf{F}}_{\rm D}(\mathbf{q}_{g}^{(i)}))
		\nonumber \\
		&~~~~~~~~~~-\eta \mathcal{T}_{2}( \mathbf{q}_{g}^{(i)} ) | \mathcal{T}_{1}( \mathbf{q}_{g}^{(i)} ) |.
	\end{align}
	Note that the first term of~\eqref{Exp:fit_fun} corresponds to the objective function in Problem (\textbf{P6}), and the second term is the penalty function introduced to ensure constraint~\eqref{C:XL-mMA1}. Specifically, for any position vector $\mathbf{q}$, the set $\mathcal{T}_{1}( \mathbf{q})$ is defined as the collection of all element pairs within $\mathbf{q}$ that violate constraint~\eqref{C:XL-mMA1}, which is given by 
		$\mathcal{T}_{1}( \mathbf{q} ) =\big\{({q}_{i},{q}_{j}) \big| |{q}_{i}-{q}_{j}|< \tilde{N} d,  \forall i \neq j\in \mathcal{M} \big\}.$
	$\mathcal{T}_{2}( \mathbf{q} )$ characterizes the extent of constraint violation for $\mathbf{q}$, which is defined as
		$\mathcal{T}_{2}( \mathbf{q}) = \sum_{ ({q}_{i},{q}_{j}) \in \mathcal{T}_{1}( \mathbf{q} ) } (\tilde{N}d - |{q}_{i}-{q}_{j} |  ). $
	Additionally, $| \mathcal{T}_{1}( \mathbf{q}_{g}^{(i)} ) |$ represents the cardinality  of $ \mathcal{T}_{1}( \mathbf{q}_{g}^{(i)} )$ and $ \eta>0 $  is a scaling factor.
	
	\textbf{DE operations}:
	Next, the following three operations, namely, \emph{mutation}, \emph{crossover}, and \emph{selection}, are performed to update the individuals (i.e., candidate positions of subarrays) in each DE iteration. 
    
    \emph{Mutation}: Generate a mutant vector $\mathbf{v}_{g}^{(i)}$ via
		\begin{align}\label{Exp:mut}
			\mathbf{v}_{g}^{(i)} = \mathbf{q}_{\rm best}^{(i-1)}+ F(\mathbf{q}_{r_{1}}^{(i-1)}-\mathbf{q}_{r_{2}}^{(i-1)}),
		\end{align}
		where $\mathbf{q}_{r_{1}}^{(i-1)}$ and $\mathbf{q}_{r_{2}}^{(i-1)}$ are distinct vectors randomly selected from  $\mathcal{G}_{\rm t}^{(i-1)}$ with $r_{1}$ and $r_{2}$ satisfying $r_{1} \neq r_{2} \neq g, \forall g\in \mathcal{G}$. Additionally, $\mathbf{q}_{\rm best}^{(i-1)}$ refers to the position vector with a maximum value of fitness function, and $F$ is the mutation factor for controlling the global search capability and the convergence rate of DE algorithm.
        
		\emph{Crossover}: After mutation operation, the crossover operation is to produce a trial vector $\mathbf{u}_{g}^{(i)} $ via genetic exchanges between the mutant individual $\mathbf{v}_{g}^{(i)}$ and current individual $ \mathbf{q}_{g}^{(i-1)} $. 
		The process of crossover operation can be mathematically expressed as
		\begin{equation}\label{Exp:cro1}
			\mathbf{u}_{g}^{(i)} = \left\{
			\begin{aligned}
				&\mathbf{v}_{g}^{(i)},~~~\textrm{if}~ {\rm rand}(1) < C_{R}~\textrm{or}~g =G_{\rm rand}^{(i)}, \\
				&\mathbf{q}_{g}^{(i-1)},~\textrm{otherwise},
			\end{aligned}
			\right.
		\end{equation}
		where  $C_{R} \in [0,1]$ is a crossover factor, $G_{\rm rand}^{(i)}$ is a random index for ensuring at least one crossover, and  $ {\rm rand}(1) \sim \mathcal{U}(0,1)$ represents a random variable. To satisfy the boundary constraint~\eqref{C:XL-mMA2}, the trial vectors are constrained within $\mathcal{C}_{\mathbf{q}}$, which is mathematically expressed as
        \begin{align}\label{Exp:cro2}
        \big[\mathbf{u}_{g}^{(i)}\big]_{m}=\max\{
        \min\{\big[\mathbf{u}_{g}^{(i)}\big]_{m},q_{\max}\},q_{\min}
        \big\}.
        \end{align}

		 \emph{Selection}: After crossover operation, the trial position vectors are evaluated using the fitness function to generate the next population, which is determined by
		\begin{equation}\label{Exp:sel} 
			\mathbf{q}_{g}^{(i)}=\left\{
			\begin{aligned}
				&\mathbf{u}_{g}^{(i)}, &&\textrm{if}~~ \mathcal{F} ( \mathbf{u}_{g}^{(i)})> \mathcal{F} ( \mathbf{q}_{g}^{(i-1)} ),  \\
				&\mathbf{q}_{{g}}^{(i-1)}, &&\textrm{otherwise}.  \\
			\end{aligned}
			\right.
		\end{equation}	

	After $I$ iterations, a suboptimal solution to Problem (\textbf{P6}), denoted by $\mathbf{q}^{*}$, is obtained, i.e., $\mathbf{q}^{*} = \mathbf{q}_{\rm best}^{(I)}$.
	
\begin{remark}[Algorithm convergence and computational complexity] \rm 
First, consider the convergence of  algorithm for solving Problem \textbf{(P4)}. For the inner problem of optimizing $\mathbf{F}_{\rm D}^{(k)}$, we have 
		$\sum_{b=1}^{B} \varpi_{b} \tilde{R}_{{\rm B},b}(\mathbf{q},\mathbf{F}_{\rm D}^{(k)}) \ge \sum_{b=1}^{B} \varpi_{b} \tilde{R}_{{\rm B},b}(\mathbf{q},\mathbf{F}_{\rm D}^{(k-1)})$,
	 since the objective function of Problem \textbf{(P8)} is non-decreasing after each iteration. Thus, the convergence of inner problem is guaranteed. For the DE algorithm to solve the outer problem, the update of $\mathbf{q}_{\rm best}^{(i)}$ over iterations satisfies
		$\mathcal{F}( \mathbf{q}_{\rm best}^{(i)}) \ge \mathcal{F}( \mathbf{q}_{\rm best}^{(i-1)})$, the convergence of proposed algorithm can be guaranteed.
	Next, the computational complexity of proposed algorithm is analyzed, which is determined by the inner problem of optimizing $\mathbf{F}_{\rm D}$ and the outer problem of updating $ \mathbf{q}$. The complexity of interior point method to solve Problem \textbf{(P8)} is in the order of $\mathcal{O}(B^6\sqrt{2(W+1)}\log(1/\xi))$, where $\xi$ is required precision of CVX solver~\cite{MyTPM}. By denoting $G I$ the number of fitness evaluations of the DE algorithm, the total computational complexity of proposed algorithm is in the order of $\mathcal{O}( I_{1} G  I B^6\sqrt{2(W+1)}\log(1/\xi))$, where $I_{1}$ represents the number of SCA iterations.
\end{remark}	

\vspace{-6pt}	
\section{Numerical Results}\label{Sec:Sim}
	In this section, numerical results are presented to showcase the efficiency of proposed movable XL-array scheme in enhancing the performance of mixed-field covert communication. 
	
	\vspace{-10pt}
	\subsection{System Setup and Benchmark Schemes}
	The system parameters are set as follows.
    The XL-array operating at $f = 30$ GHz frequency bands consists of $ N = 256$ antennas, which are divided into $ M = 8$ subarrays. For covert communications, the XL-array (Alice) serves $ B= 2 $ Bobs with the existence of $W=2$ Willies. A challenging scenario is considered, where the spatial angles of Bob $1$ and Willie $1$ are set as $\theta_{{\rm B},1} = \theta_{{\rm W},1} = 0 $ and the spatial angles of Bob $2$ and Willie $2$ are randomly and uniformly distributed as  $\theta_{{\rm W/B},2}\sim \mathcal{U}(-0.05,0.05)$. Furthermore, the ranges of Bobs and Willies are modeled as $r_{{\rm B},1}\sim\mathcal{U}(25,35)$ m, $r_{{\rm B},2}\sim \mathcal{U}(150,160)$ m, and $r_{{\rm W},w}\sim \mathcal{U}(15,25)$ m. For wireless channels, we consider the general Rician fading model. Specifically, the complex-valued path gain for the LoS path is modeled as $ g = \sqrt{\frac{\kappa}{\kappa+1}} \frac{\sqrt{\hbar}}{r} e^{-\jmath\frac{2\pi r }{\lambda}}$, where $\kappa = 10$ dB represents the Rician factor and $\hbar = ({\lambda}/{4\pi})^2$ is the reference channel gain. The complex-valued path gain for the NLoS path is modeled as $g^{(\ell)} \sim \mathcal{CN}(0, \bar{\sigma}_{\ell}^2)$, where $\bar{\sigma}_{\ell} = \sqrt{\frac{1}{\kappa+1}}\frac{\sqrt{\hbar}}{r}$~\cite{WuDFT}. The weighting factors are set proportional to ranges of Bobs for fairness comparison, which are normalized as ${\varpi_{1}} = \frac{ r_{{\rm B},1}}{r_{{\rm B},1} + r_{{\rm B},2}}$, $\varpi_{2} =  \frac{ r_{{\rm B},2}}{r_{{\rm B},1} + r_{{\rm B},2}}$~\cite{Guo2020WSR}. 
        Unless otherwise specified, other system parameters are set as $L = 4$, $ P = 1 $ W, $\rho = 3 $ dB, $ \varsigma = 0.1$, $\sigma_{{\rm B}}^2  = \tilde{\sigma}_{\rm w}^2 =-90 $ dBm, $G = 50$, $I = 50$, $F = 0.3$, $C_{R} = 0.9$, and $\eta = 1000$, $\mathcal{C}_{\mathbf{q}} = [-(N-1)d, (N-1)d] $ m.
    
	For performance comparison, the following benchmark schemes are considered. 
	\begin{itemize}
		
		\item \emph{Fixed XL-array scheme}: For this scheme, the fixed XL-array is considered, and the hybrid beamforming matrices are designed via the proposed two-stage framework to maximize the covert transmission rate.
		
	   \item \emph{Fixed XL-array scheme without  (w.o.) covertness constraint}: For this scheme, the covertness constraint (i.e.,~\eqref{C:Covert}) is not considered for fixed XL-array, while hybrid beamforming matrices are designed via the proposed two-stage framework.
		
		\item \emph{Random movable XL-array (RanM XL-array)}: For this scheme, $G$ random candidates of $\mathbf{q}$ are generated, and the best position vector~$\mathbf{q}_{\rm best}$ is selected to maximize the covert transmission rate via the two-stage framework.
		
		\item \emph{RanM XL-array w.o. covertness constraint}: Similar to \emph{RanM XL-array} scheme, this scheme aims to maximize the achievable sum-rate based on random candidates of $\mathbf{q}$ without considering covertness constraint. 
	\end{itemize}

	In addition to the proposed low-complexity two-stage hybrid beamforming method, we also adopt the FDA-based method in~\cite{Wangzl} (named as hybrid~\cite{Wangzl}) for performance evaluation.
	
	\subsection{Performance Analysis}
	 \subsubsection{Convergence of the Proposed Algorithm}
    
    \begin{figure}[t]
    	\centering
    	\begin{subfigure}{0.49\linewidth}
    		\centering
    		\includegraphics[width=1\linewidth]{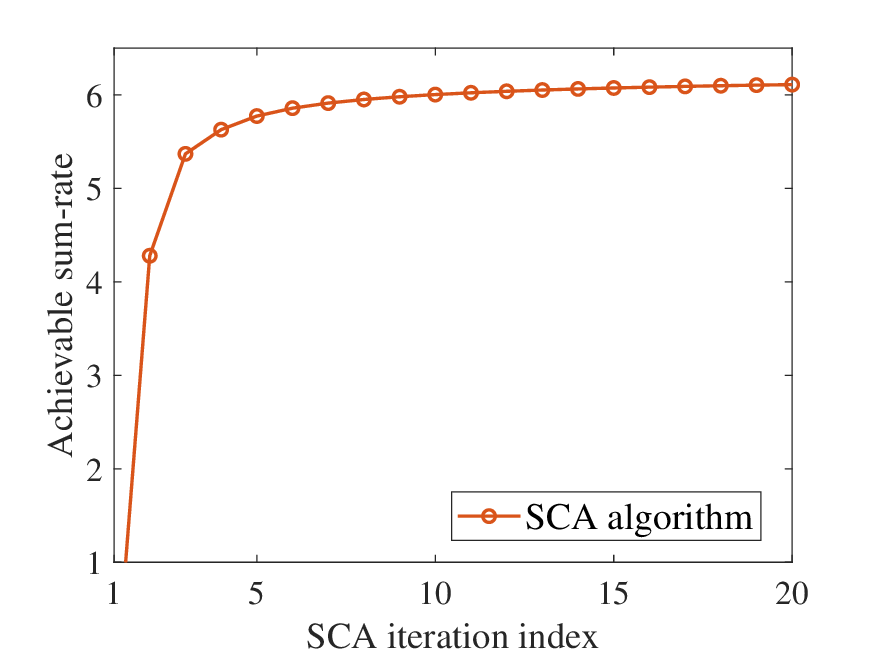}
    		\caption{Convergence of the proposed SCA technique.}
    		\label{Fig:SCA}
    	\end{subfigure}
    	\begin{subfigure}{0.49\linewidth}
    		\vspace{0em}
    		\centering
    		\includegraphics[width=1\linewidth]{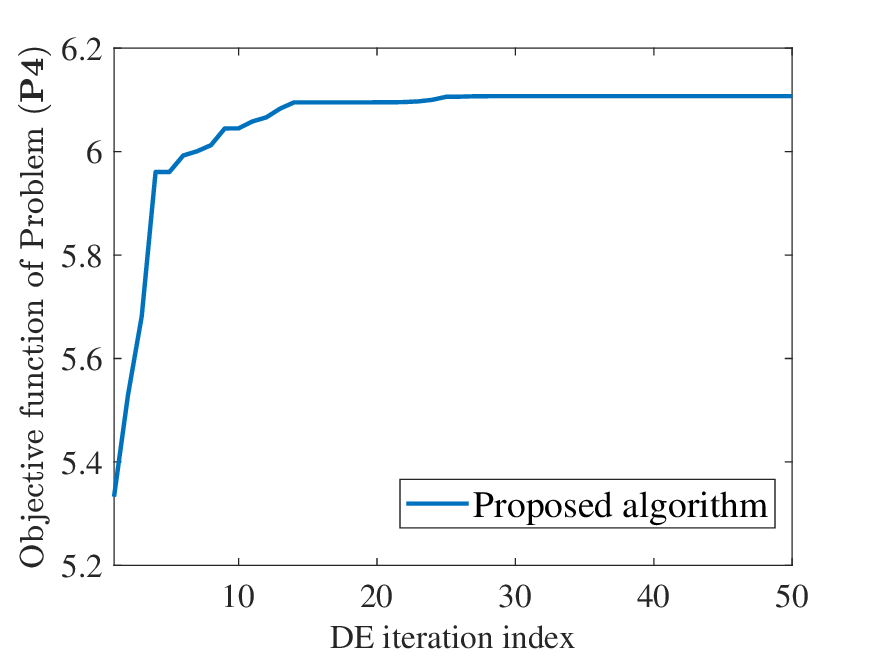}
    		\caption{Convergence of the proposed DE algorithm.}
    		\label{Fig:DE}
    	\end{subfigure}
    	\caption{Convergence of the proposed algorithm.}
    	\vspace{-20pt}
    	\label{Fig:Sec5-Convergence2}
    \end{figure}
    
    The convergence of proposed SCA technique and DE algorithm for solving Problem \textbf{(P4)} is plotted in Fig.~\ref{Fig:Sec5-Convergence2}. As shown in Fig.~\ref{Fig:Sec5-Convergence2}(a), the achievable sum-rate obtained from SCA algorithm converges to a stationary point after around  $20$ iterations. Additionally, Fig.~\ref{Fig:Sec5-Convergence2}(b) shows that the proposed DE Algorithm converges to a stationary point within about 30 DE iterations.
	
	\subsubsection{Effect of Covertness Constraint}

	To show the effect of covertness constraint, we plot in Fig.~\ref{Fig:Sec5-Rate_vs_MDEP} the curves of achievable sum-rate versus the covertness constant $\varsigma$ by different schemes, where a smaller $\varsigma$ incurs a more stringent covert transmission condition.  
	Several key observations are made as follows. First, the achievable sum-rates of all schemes increase with $\varsigma$, because a more relaxed covertness constraint allows for a more effective beamforming design and power allocation.
	Second, the achievable sum-rates of all schemes
	with FDA-based hybrid beamforming method~\cite{Wangzl} exhibit significantly better covert rate performance than the RanM XL-array and Fixed XL-array schemes, owing to its effective beamforming design. 
	Third, the proposed two-stage beamforming design achieves close rate performance with the FDA-based method, even under the highly stringent covertness constraint (i.e., $\varsigma = 0.02$). This phenomenon confirms the efficacy of movable XL-array in reducing channel correlation, thereby enabling the two-stage method to achieve high-quality solutions.
	
	\subsubsection{Effect of Maximum Transmit Power}
    In Fig.~\ref{Fig:Sec5-Rate_vs_P}, we plot the achievable sum-rate versus the maximum transmit power.  
    Among all schemes without considering the covertness constraint, the proposed movable XL-array scheme achieves a rate performance close to its upper bound (i.e., both inter-user interference and the covertness constraint are neglected). This observation indicates that position optimization effectively reduces the channel correlations among mixed-field Bobs, thereby improving the achievable sum-rate compared with other schemes. Next, for all schemes accounting for the covertness constraint, a similar phenomenon is observed that the proposed scheme achieves the highest covert transmission rate. This is attributed to the effectiveness of the movable XL-array in reducing channel correlations between Bobs and Willies. Additionally, compared with the RanM XL-array scheme, the proposed scheme yields a higher covert transmission rate, further demonstrating the advantage of proposed position optimization algorithm.
    Interestingly, for the RanM XL-array and fixed XL-array schemes, the achievable sum-rate remains largely unchanged as the maximum transmit power increases, because a higher transmit power used for signal transmission will violate the covertness constraint. Notably, the proposed scheme exhibits a trend similar to schemes without covertness constraint, indicating that position optimization significantly reduces channel correlations between Willies and Bobs.
    \begin{figure}[t]
    	\centering
    	\includegraphics[width=0.35\textwidth]{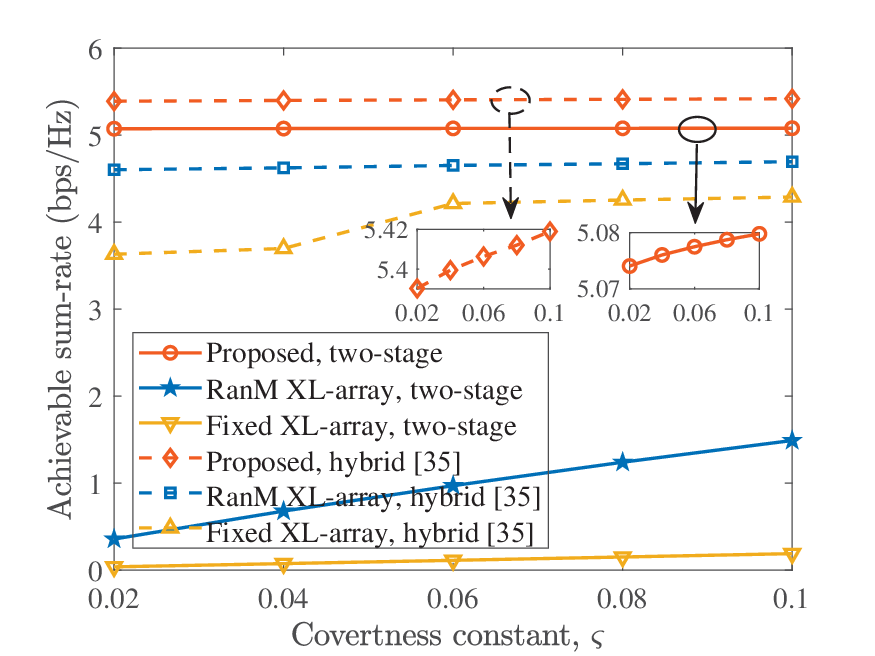}
    	\caption{Achievable sum-rate versus covertness constant.} \label{Fig:Sec5-Rate_vs_MDEP}
    	\vspace{-14pt}
    \end{figure}
    
	\begin{figure}[t]
		\centering
		\includegraphics[width=0.35\textwidth]{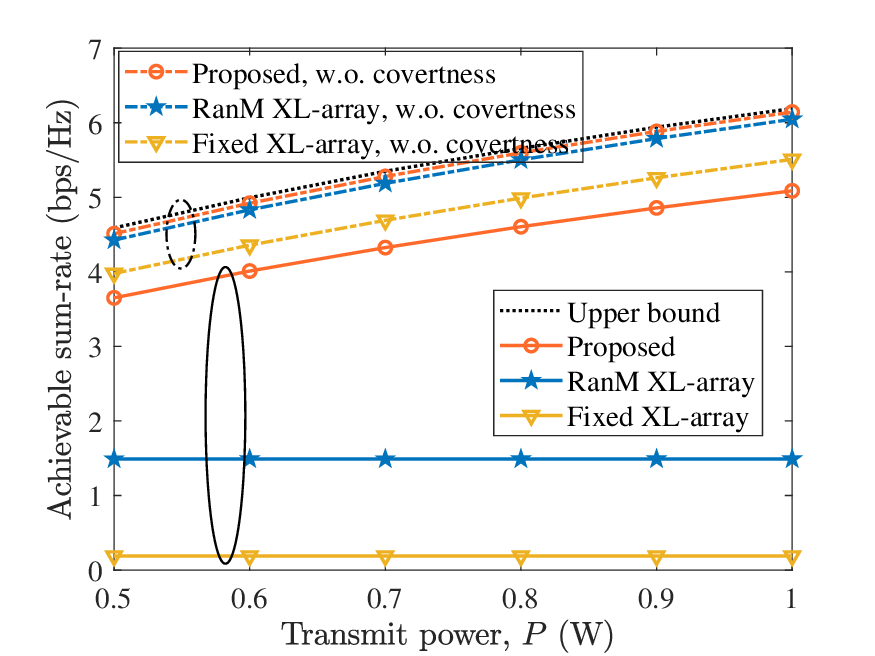}
		\caption{Achievable sum-rate versus maximum transmit power.} \label{Fig:Sec5-Rate_vs_P}
		\vspace{-14pt}
	\end{figure}

    \subsubsection{Effect of Number of Transmit Antennas}
	\begin{figure}[t]
		\centering
		\includegraphics[width=0.35\textwidth]{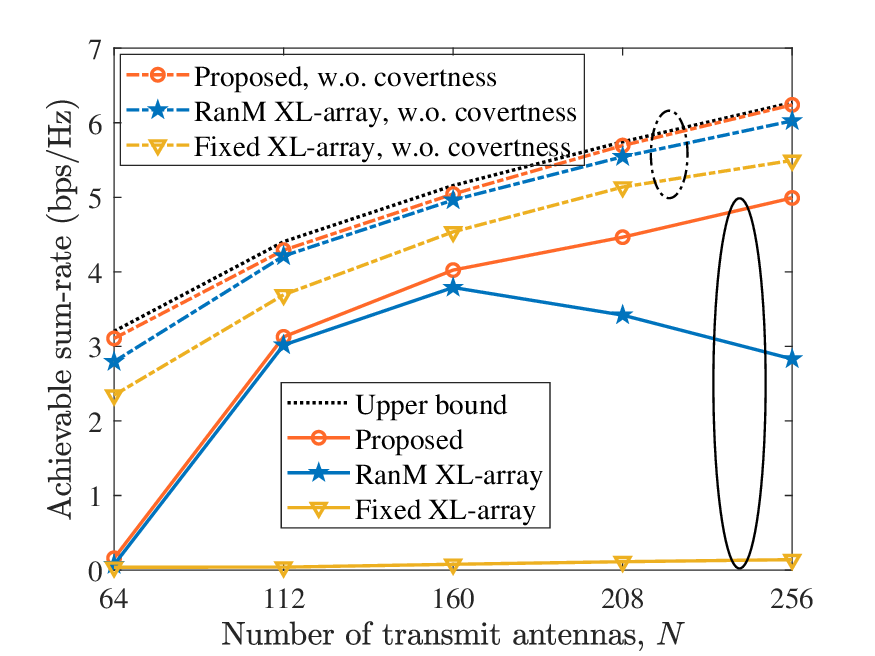}
		\caption{Achievable sum-rate versus number of antennas.} \label{Fig:Sec5-Rate_vs_N}
		\vspace{-14pt}
        \end{figure}
    In Fig.~\ref{Fig:Sec5-Rate_vs_N}, we show the performance of different schemes versus the number of XL-array antennas $(N)$. 
    It is observed that the achievable sum-rates of all schemes (except RanM XL-array) increase with the number of transmit antennas. However, when $N=64$, the covertness constraint significantly limits the achievable sum-rate, since all near-field channels become far-field ones when $N=64$, hence exhibiting high channel correlation at the same spatial angle. Moreover, as $N$ increases, the performance gap between the RanM XL-array and the proposed scheme becomes more pronounced due to the more significant mixed-field effect.
    The RanM XL-array scheme fails to effectively reduce strong channel correlation due to randomly generated position vectors.
    In contrast, the proposed scheme achieves a satisfactory covert transmission rate by utilizing the additional DoF offered by MA and effective position optimization by the proposed algorithm.
     

    \subsubsection{Effect of imperfect CSI for Willies}
	\begin{figure}[t]
	\centering
	\includegraphics[width=0.35\textwidth]{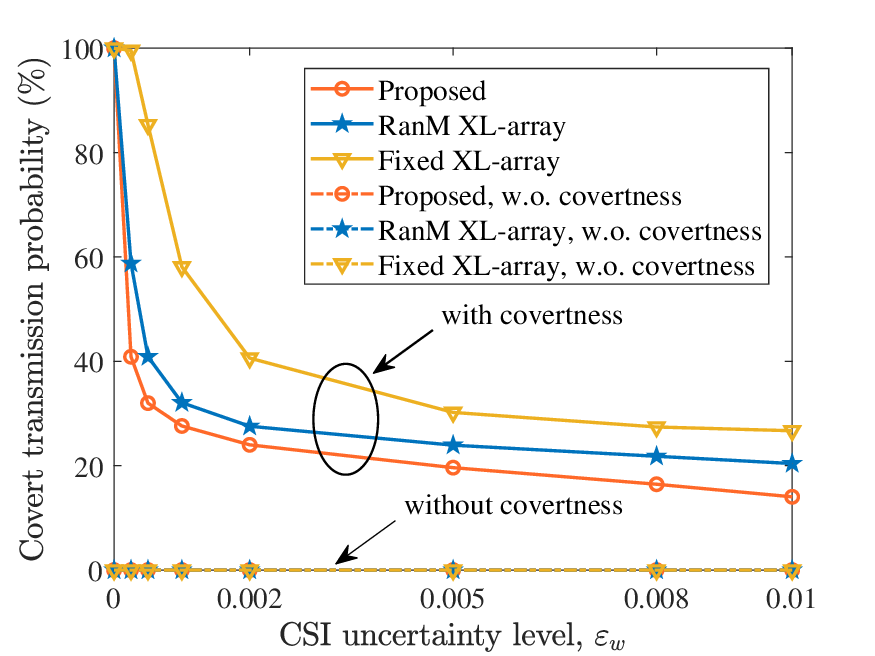}
	\caption{Covert transmission probability versus CSI error.} 
	\label{Fig:Sec5-RatevsCSIerror}
	\vspace{-14pt}
	\end{figure}
	Last, we investigate the impact of CSI errors for Willies on covert transmission probability. Specifically, the actual channel $\!\mathbf{h}_{{\rm W},w}^H \!$ is modeled~as
\begin{align}
	\mathbf{h}_{{\rm W},w}^H = \hat{\mathbf{h}}_{{\rm W},w}^H + \Delta  \mathbf{h}_{{\rm W},w}^H, w \in \mathcal{W},
\end{align}
where $\hat{\mathbf{h}}_{{\rm W},w}^H$ is the estimated channel used for designing the hybrid beamforming and the positions of movable subarrays. Additionally,  $\Delta \mathbf{h}_{{\rm W},w}^H $ is the channel estimation error, which is assumed to follow the distribution
$	\Delta  \mathbf{h}_{{\rm W},w} \sim \mathcal{CN}(\mathbf{0},\boldsymbol{\Sigma}_{w}),~ \boldsymbol{\Sigma}_{w} = \varepsilon_{w}^2 \|\hat{\mathbf{h}}_{{\rm W},w}\|_{2}^{2} \mathbf{I}_{N}$,
where $\varepsilon_{w}\in [0,1]$ denotes the CSI uncertainty level~\cite{MyTWC}.

We present in Fig.~\ref{Fig:Sec5-RatevsCSIerror} the covert transmission probabilities (i.e., $\text{Pr}\{f_{{\rm Co},w} \le \tilde{\sigma}_{\rm w}^2 (\frac{\rho^{2\varsigma} - 1}{\rho}),\Delta  \mathbf{h}_{{\rm W},w} \!\sim\!\mathcal{CN}(\mathbf{0},\boldsymbol{\Sigma}_{w}), \forall w \in \mathcal{W}\}$) of different schemes versus the CSI uncertainty level for Willies. When $\varepsilon_{w}$ increases from $0$ to $0.01$, all schemes without considering the covertness constraint~\eqref{C:Covert} suffer a covert transmission probability of zero. In contrast, for other  schemes accounting for covertness, their covert transmission probability decreases with $\varepsilon_{w}$. These results highlight that, in scenarios with CSI errors, robust beamforming design is crucial to ensure covertness.

%
%
	
	\vspace{-8pt}
       \section{Conclusions}\label{Sec:Con}
        In this paper, we proposed to deploy a \emph{movable XL-array} at the BS (Alice) to enhance mixed-field covert communication with multiple near- and far-field Bobs and near-field Willies. We revealed that the conventional fixed XL-array suffers degraded rate performance due to the energy-spread effect, which can be improved by movable XL-array. Moreover, we also characterized the covert transmission conditions for both near-field and far-field Bobs. 
        Based on these insights, we formulated an optimization problem to maximize the achievable sum-rate under covertness constraint. To solve this non-convex problem, we designed a two-layer optimization framework: an inner problem to optimize hybrid beamforming and an outer problem to optimize subarray movement. Numerical results validated that the proposed movable XL-array is capable of enhancing covert rate performance in mixed-field scenarios.

    \vspace{-6pt}
	\begin{appendices}
		\section{}\label{App:Modular_approx}
		Let   $ \Omega_{\theta} = \theta_{2} - \theta_{1} $, $ \Omega_{r} = \frac{1-\theta_{1}^2}{2 r_{1}} -\frac{1-\theta_{2}^2}{2 r_{2}} $.
		By substituting~\eqref{Exp:antenna_position} into~\eqref{Exp:Corfun}, we have
		\begin{align}\label{Exp:modular_expan}
			&\chi(\mathbf{q},\{\theta_1,r_1\},\{\theta_{2},r_{2}\})  \nonumber \\
			= & \frac{1}{N}  \Big| \sum_{m=1}^{M}  e^{\jmath\frac{2\pi}{\lambda} ({q}_{m} \Omega_{\theta} + {q}_{m}^2 \Omega_{r} )  } \sum_{\tilde{n}=1}^{\tilde{N}} \big( e^{\jmath\frac{2\pi}{\lambda} ( \frac{2\tilde{n}-\tilde{N}-1}{2} d  \Omega_{\theta}  ) }  \nonumber \\
			&\cdot e^{\jmath\frac{2\pi}{\lambda} (2 {q}_{m}  \frac{2\tilde{n}-\tilde{N}-1}{2} d \Omega_{r})  }
			\cdot e^{\jmath\frac{2\pi}{\lambda} (\frac{2\tilde{n}-\tilde{N}-1}{2} d)^2 \Omega_{r}  }
			\big)
			\Big|.
		\end{align}
		When $r_{1}$ and $r_{2}$ are larger than $\frac{2(\tilde{N}d)^2}{\lambda}$, the phase caused by the second-order term can be ignored. By denoting   $\Omega_{m} = 2 {q}_{m} \Omega_{r} + \Omega_{\theta} $, \eqref{Exp:modular_expan} can be approximated as
		\begin{align}
			& \frac{1}{N}\Big| \sum_{m=1}^{M} e^{\jmath\frac{2\pi}{\lambda} ({q}_{m} \Omega_{\theta} + {q}_{m}^2 \Omega_{r} )  } \sum_{\tilde{n}=1}^{\tilde{N}}  e^{\jmath\frac{2\pi}{\lambda} \frac{2\tilde{n}-\tilde{N}-1}{2} d  (\Omega_{\theta} + 2 {q}_m \Omega_{r}  ) }
			\Big| \nonumber \\
			=&  \frac{1}{N}\Big| \sum_{m=1}^{M} e^{\jmath\frac{2\pi}{\lambda} ({q}_{m} \Omega_{\theta} + {q}_{m}^2 \Omega_{r} )  } \frac{\sin(\frac{\tilde{N}\pi \Omega_{m}}{2})}{\sin(\frac{\pi \Omega_{m}}{2})}
			\Big|,
		\end{align}
		thus completing the proof.
		\vspace{-10pt}
		\section{}\label{App:CT_region}
		Based on Definition~\ref{Def:2}, the correlation $\chi_{{\rm W}_{1},{\rm B}_{1} } $ in~\eqref{Exp:CorW1B1_condition} can be approximated as  $\chi_{{\rm W}_{1},{\rm B}_{1} } \approx |G(\beta^{({\rm W}_1,{\rm B}_1)} )|$, where  $\beta^{({\rm W}_1,{\rm B}_1)}$ is given by
		$	\beta^{({\rm W}_1,{\rm B}_1)} = \sqrt{\frac{N^2d^2 (1-\theta_{{\rm W},1}^2) }{2\lambda}\left|\frac{1}{r_{{\rm W},1}}-\frac{1}{r_{{\rm B},1}} \right| }$.
		To ensure~\eqref{Exp:CorW1B1_condition}, the following inequality 
			$|G(\beta^{({\rm W}_1,{\rm B}_1)} )| \le | G( \beta_{{\Delta}_{1}(\epsilon)} )  |$
		should be guaranteed, where $ \beta_{\Delta} $ denotes a nonnegative constant that satisfies $|G(\beta_{\Delta})| = \Delta $. Considering that $|G(\beta)|$ shows a significant downward trend with minor fluctuations~\cite{Cui2022CE} as $\beta$ increases, 
		the covert transmission condition in~\eqref{Exp:CorW1B1_condition} can be approximated as $\beta^{({\rm W}_1,{\rm B}_1)} \ge \beta_{{\Delta}_{1}(\epsilon)}.$
		By denoting  $\Pi = \frac{2\lambda (\beta_{{\Delta}_{1}(\epsilon)})^2}{N^2d^2 (1-\theta_{{\rm W},1}^2)}$, the above inequality can be rewritten as $\{r_{{\rm B},1}\big| \Pi\le \big|  \frac{1}{r_{{\rm W},1}}-\frac{1}{r_{{\rm B},1}} \big|\}$, 
		thus completing the proof.
		
	\end{appendices}

	\bibliographystyle{IEEEtran}
	\bibliography{Ref_HFCovertcom.bib}
	
\end{document}